\documentclass[journal,10pt,comsoc]{IEEEtran}
\usepackage{amsthm}
\usepackage{amssymb}
\usepackage{amsmath}
\usepackage{amsfonts}
\usepackage{graphicx}
\usepackage{algorithm}
\usepackage{algorithmic}
\usepackage{epstopdf}
\usepackage{cite}
\usepackage{hyperref}
\usepackage{subfigure}
\usepackage{color}
\newtheorem{theorem}{\textbf{Theorem}}

\newtheorem{corollary}{\textbf{Corollary}}

\newtheorem{lemma}{\textbf{Lemma}}

\newtheorem{proposition}{\textbf{Proposition}}
\newtheorem{remark}{\textbf{Remark}}

\begin{document}
\title{Optimization and Analysis of Probabilistic Caching in $N$-tier Heterogeneous Networks}
\author{Kuikui Li, Chenchen Yang, Zhiyong Chen,~\IEEEmembership{Member,~IEEE,} and Meixia Tao,~\IEEEmembership{Senior Member,~IEEE}
\thanks{Manuscript received December 6, 2016; revised April 27, 2017 and September 15, 2017; accepted November 15, 2017. This work is supported by the National Natural Science Foundation of China under Grants 61571299, 61671291, 61329101, 61528103, and 61521062. This work was presented in part at the IEEE SPAWC 2017 \cite{kkl}. The associate editor coordinating the review of this paper and approving it for publication was Prof. Chan-Byoung Chae. (\emph{Corresponding author: Zhiyong Chen.})}
\thanks{The authors are with the Department of Electronic Engineering, Shanghai Jiao Tong University, Shanghai,
200240, P. R. China. Z. Chen is also with the Cooperative Medianet Innovation
Center, Shanghai, China. M. Tao is also with Shanghai Institute for Advanced Communication and Data Science (email: kuikuili@sjtu.edu.cn; zhanchifeixiang@sjtu.edu.cn; zhiyongchen@sjtu.edu.cn; mxtao@sjtu.edu.cn.).}}
\maketitle
\begin{abstract}
In this paper, we study the probabilistic caching for an $N$-tier wireless heterogeneous network (HetNet) using stochastic geometry. A general and tractable expression of the successful delivery probability (SDP) is first derived. We then optimize the caching probabilities for maximizing the SDP in the high signal-to-noise ratio (SNR) regime. The problem is proved to be convex and solved efficiently. We next establish an interesting connection between $N$-tier HetNets and single-tier networks. Unlike the single-tier network where the optimal performance only depends on the cache size, the optimal performance of $N$-tier HetNets depends also on the BS densities. The performance upper bound is, however, determined by an equivalent single-tier network. We further show that with uniform caching probabilities regardless of content popularities, to achieve a target SDP, the BS density of a tier can be reduced by increasing the cache size of the tier when the cache size is larger than a threshold; otherwise the BS density and BS cache size can be increased simultaneously. It is also found analytically that the BS density of a tier is \emph{inverse} to the BS cache size of the same tier and is \emph{linear} to BS cache sizes of other tiers.
\end{abstract}
\section{Introduction}
The global mobile data traffic is estimated to increase to 30.6 exabytes per month by 2020, an eightfold growth over 2015, and the contribution by video is foreseen to increase from $55\%$ in 2015 to $75\%$ in 2020 \cite{cisco2016global}. To address this mobile data tsunami and hence meet the capacity requirement for the future 5G network\cite{6824752}, an effective and promising candidate solution is to deploy a dense network with heterogeneous base stations (BSs), such as macro BSs, relays, femto BSs and pico BSs \cite{multitier}. The heterogeneous network (HetNet) can provide higher throughput and spectral efficiency. In the meantime, it also faces two challenges. One is the tremendous burden on the backhual link due to the explosive demand for video contents during the peak time and the other is high CAPEX and OPEX due to the denser BSs.

Recently, 
caching popular contents at BSs has been introduced as a promising technique to offload mobile data traffic in cellular networks \cite{5G1,5G2}. Unlike the communication resources, the storage resources are abundant, economical, and sustainable. By exploiting the abundance of the storage resources in wireless networks, significant gains in network capacity through caching can be expected \cite{maddah2014fundamental}, which enables caching to be an essential functionality of emerging wireless networks \cite{ThreeColors}.

The aim of this work is to study how caching can address the aforementioned challenges in a multi-tier HetNet. In specific, we first would like to find out what is the optimal cache placement strategy in order to alleviate the traffic burden in backhaul links to the minimum. Second, we would like to find out if the deployment cost of dense BSs can be traded by BS cache storage, and if so, what are the tradeoffs and what conditions must be met in order for it to happen.
\subsection{Related Work}
Caching has the potential to alleviate the heavy burden on the capacity-limited backhaul link and also improves user-perceived experience \cite{content}. 
Utilizing the tool of stochastic geometry \cite{andrews2011tractable}, the work \cite{bacstug2015cache} formulates the caching problem in a scenario where small BSs are distributed according to a homogeneous Poisson Point Process (HPPP). 
The authors in \cite{tamoor2016caching} consider a two-tier HetNet and derive a closed-form expression of the outage probability by jointly considering spectrum allocation and storage constraints. In \cite{yang2015analysis}, the authors consider a 3-tier HPPP-based HetNet with caching and theoretically elaborate the average ergodic rate, outage probability and delay. Considering an HPPP-based cache-enabled small cell network, a closed form expression of the outage probability and the optimal BS density to achieve a target hit probability are derived in \cite{hassan2015modeling}. The work \cite{7880694} proposes a cluster-centric small cell network and designs cooperative transmission scheme to balance transmit diversity and content diversity. It is worth noting that these works mainly focus on the performance analysis of cache-enabled wireless networks for given caching strategies, such as caching the most popular contents. 

Caching strategy is an important issue for cache-enabled wireless networks.
Previous works on the optimal caching strategy design can be classified into two trends based on whether channel fading and interference are considered. The early trend focuses on the connection topology only while ignoring channel fading and interference. The authors in \cite{femto} formulate a cache placement problem in distributed helper stations to minimize the average download delay with both uncoded and coded caching. It is shown in \cite{femto} that the optimal caching problem in a wireless network with fixed connection topology is an NP-hard problem (without coding).
In \cite{poularakis2014approximation}, a joint routing and caching design problem is studied to maximize the content requests served by small BSs. By reducing the NP-hard optimization problem to a variant of the facility location problem, algorithms with approximation guarantees are established. The second trend takes into account channel fading and interference for caching optimization by mostly utilizing the tool of stochastic geometry.
The work \cite{blaszczyszyn2015optimal} proposes an optimal randomized caching policy to maximize the total hit probability and overviews different coverage models to evaluate the performance. The works \cite{7562510,TWCCaching2016,cui2015analysis} optimize the probabilistic caching strategy to maximize the successful download probability in small cell networks. Further, a closed-form expression for the optimal caching probabilities is obtained in the noise-limited scenario in \cite{TWCCaching2016}. In \cite{song2016optimal}, a greedy algorithm is proposed to find the optimal caching strategy to minimize the average bit error rate. The work \cite{7485844} studies the problem of joint caching, routing, and channel assignment for video delivery over coordinated multicell systems of the future Internet.

Recently, caching strategy optimization is extended to wireless heterogeneous networks. The combination of the optimal caching and the network heterogeneity brings more gains in network capacity. Utilizing the tool of stochastic geometry, the works \cite{7723871,7875124} investigate the optimal probabilistic caching at helper stations while assuming deterministic caching at macro stations to maximize the successful transmission probability in a two-tier HetNet. 
Based on \cite{blaszczyszyn2015optimal}, the work \cite{serbetci2016optimal} considers different types of BSs with different cache capacities. The cache optimization problem for the first type of BSs is solved by assuming that the
placement strategy for other types of BSs is given. The joint probabilistic caching optimization problem for all types of BSs is yet not considered, and little analytical insight on the cache design and
system performance is available. In general, the joint optimization for probabilistic cache placements in different tiers of a HetNet is very challenging due to the different tier association probabilities brought by the content diversity as well as the complicated interference distribution by the nature of network heterogeneity.

Furthermore, 
a tradeoff between the small BS density and total storage size is firstly presented in \cite{bacstug2015cache}, where each small BS caches the most popular contents. Using the optimal caching scheme, \cite{7875124} shows that the helper density can be traded by the cache size to achieve a target area spectral
efficiency. Note that the tradeoff studies in \cite{bacstug2015cache,7875124} are conducted numerically only without theoretical analysis. Deriving and analyzing the tradeoff theoretically has not been solved. In \cite{HowMuchCaching}, the authors address the question that how much caching is needed to achieve the linear capacity scaling in the dense wireless network based on scaling law method.
\subsection{Contributions}
In this work, we first investigate the optimal probabilistic caching to maximize the successful delivery probability (SDP) in a general $N$-tier ($N\ge2$) wireless cache-enabled HetNet. We next establish an interesting connection between $N$-tier HetNets and single-tier networks. We then address the tradeoffs between the BS caching capability and the BS density analytically based on the uniform caching strategy.
The main contributions are summarized as follows:
\begin{itemize}
\item\emph{Analyzing and optimizing the SDP for the $N$-tier HetNet:}
   We derive the tier association probability and the SDP by modeling the BS locations in the HetNet as $N$-tier independent HPPPs. The optimal probabilistic caching problem for maximizing the SDP is then formulated. We prove that this problem is concave in the high signal-to-noise ratio (SNR) regime. The sufficient and necessary conditions for the optimal solution are derived.
\item\emph{Highlighting the connection between $N$-tier HetNets and single-tier networks:} We further study the optimal caching problem in special cases, and find that the maximum SDP of single-tier networks only depends on the cache size while that of $N$-tier HetNets is also determined by the BS densities and transmit powers.
    Moreover, in the high SNR regime, we prove that there exists a single-tier network such that the maximum SDP of the $N$-tier HetNet is upper bounded by that of the single-tier network. When all tiers of BSs have the same cache size, the $N$-tier HetNet performs the same as the single-tier network, regardless of the network heterogeneity.
\item\emph{Presenting insights on the impacts of the key network parameters:} We first show that the optimal performance of single-tier networks is independent of the BS density and transmit power. Then, under uniform caching strategy, we analytically present the impacts of the BS cache size, density and transmit power of each tier on the system performance. Numerical results also verify our analytical results.
\item\emph{Revealing the tradeoffs between the BS density and the BS cache size:} With uniform caching strategy, our analysis reveals that, to maintain a target SDP, the network parameters are related as follows: $1)$ increasing the BS caching capability can reduce the BS density when its cache size is larger than a threshold $Q_{e}$; $2)$ the BS density $\lambda_{i}$ is \emph{\textbf{inversely proportional}} to the cache size $Q_{i}$ in the same tier, i.e., $\lambda_{i}=\frac{K_{1}}{Q_{i}-Q_{e}}$. Here, $i$ denotes the index of the tier and $K_{1}$ is independent of $\lambda_{i}$ and $Q_{i}$. For the different tiers, we prove that the BS density $\lambda_{j}$ is a \emph{\textbf{linear function}} of the cache size $Q_{i}$, i.e., $\lambda_{j}=\frac{K_{3}-K_{4}Q_{i}}{K_{5}}$, for $i\neq j$, where $K_{3}$, $K_{4}$ and $K_{5}$ are independent of $\lambda_{j}$ and $Q_{i}$. Likewise, we reveal the similar tradeoffs between the BS transmit power and the BS cache size.
\end{itemize}

The rest of this paper is organized as follows. Section \ref{systemmodel} presents the system model. The performance metric is analyzed in Section \ref{performancemetric}. In Section \ref{problem}, we formulate and solve the optimal caching problem. Then, the impacts and tradeoffs of the network parameters are shown in section \ref{tradeoffimpacts}. The numerical and simulation results are presented in Section \ref{simulationresults}, and the conclusions are drawn in Section \ref{conclusion}.
\section{System Model} \label{systemmodel}
\subsection{Network and Caching Model}


\begin{figure}[t]
\centering
\includegraphics[width=3.7in, height=1.8in]{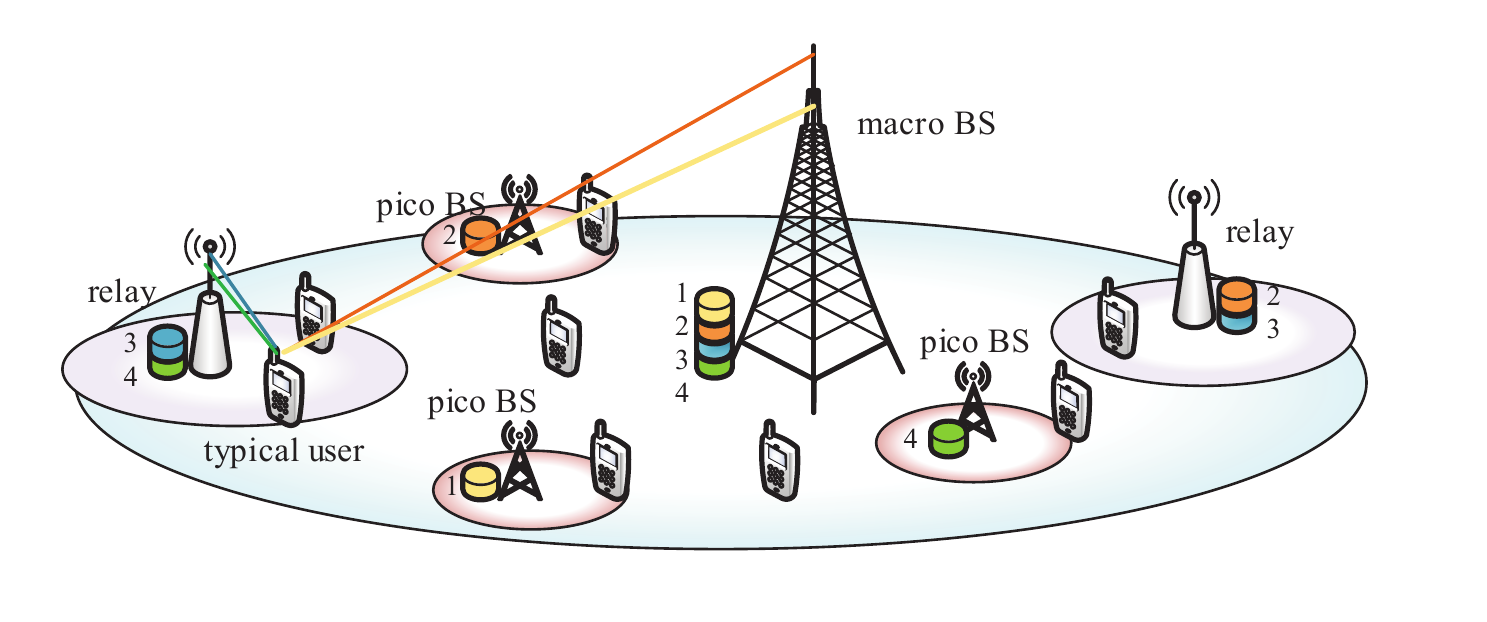}
\vspace{-7mm}
\caption{An example of a 3-tier wireless cache-enabled HetNet: A typical user can obtain contents \{3, 4\} from relay and contents \{1, 2\} from macro BS.}
\vspace{-4.5mm}
\label{SigModel}
\end{figure}
We consider a general wireless cache-enabled HetNet consisting of $N$ tiers of BSs, where the BSs in different tiers are distinguished by their transmit powers, spatial densities, biasing factors, and cache sizes\footnote{These model assumptions indicate that BSs in different tiers have different traffic load to handle, and also reflect the demand heterogeneity in different BSs.}. The locations of BSs in each tier are spatially distributed according to an independent HPPP, denoted as $\Phi_i$ with density $\lambda_i$ for $i\in\mathcal{N}\triangleq\{1, 2,\cdots, N\}$. A three-tier HetNet including macro BSs, relays and pico BSs, is illustrated in Fig. \ref{SigModel}. Consider the downlink transmission. Time is divided into discrete slots with equal duration and we study one slot of the system. For the wireless channel, both large-scale fading and small-scale fading are considered. The large-scale fading is modeled by a standard distance-dependent path loss attenuation with path loss exponent $\beta$. The Rayleigh fading channel $h$ is considered as the small-scale fading, i.e., $h\sim \mathcal{CN}(0,1)$. Each user receiver experiences an additive noise that obeys zero-mean complex Gaussian distribution with variance $\sigma^2$.

Consider a database consisting of $M$ contents denoted by $\mathcal{M}\triangleq\{1,2,\cdots,M\}$, and all the contents are assumed to have equal length\footnote{Note that the extension to the general case where contents are of different lengths is quite straightforward since the contents can be divided into chunks with equal size.}. Each user only requests one single content at each time slot. The content popularity distribution is identical among all users, represented by $\mathbf{T}\triangleq(t_{i})_{i\in\mathcal{M}}$, where each user requests the $i$-th content with probability $t_{i}$ and $\sum_{i=1}^{M}{t_{i}}=1$. The content popularity $\mathbf{T}$ is assumed to be known a prior for cache placement. Without loss of generality, we assume $t_{1}\ge{t_{2}}\ge\dots\ge{t_{M}}$. Each BS is equipped with a cache storage. The cache capacities of $N$-tier BSs are denoted as $\mathcal{Q}=\{Q_1,Q_2,\cdots,Q_N\}$, where each BS in the $i$-th tier can store at most $Q_i$ ($Q_i<M$) contents.

When a user submits a content request, the content will be delivered directly from the local cache of a BS that has cached it. If the content is not cached in any BS, it will be downloaded from the core network through backhaul links. Since the main purpose of this paper is to optimize the caching strategy to offload the backhaul traffic, we only consider the transmission of the cached contents at BSs, same as \cite{7562510}.

We adopt the probabilistic caching strategy and assume all the BSs in a same tier use the same caching probabilities. Each BS caches contents with the given probabilities independently of other BSs. Define $\mathbf{P}\triangleq[p_{ij}]_{N\times{M}}$ as the caching probability matrix where $p_{ij}$ denotes the probability that the BSs in the $i$-th tier caches the $j$-th content. It must satisfy
\begin{align}
&0\le{p_{ij}}\le{1},~~\forall{i}\in{\mathcal{N}}, \forall{j}\in{\mathcal{M}} \label{condition1}\\
&\sum_{j=1}^{M}{p_{ij}}\le{Q_i},~~\forall{i}\in{\mathcal{N}}.         \label{condition2}
\end{align}

Note that the conditions (\ref{condition1}) and (\ref{condition2}) are sufficient and necessary for the existence of a random content placement policy requiring no more than $Q_i$ slots of storage at each BS in tier $i$ for $\forall{i}\in{\mathcal{N}}$ \cite{blaszczyszyn2015optimal}. Also note that if a BS realizes the caching strategy by caching each file at random with the given probability but independently of other files, the actual cache memory in the BS can be exceeded or wasted. To strictly meet the instantaneous cache size constraint (\ref{condition2}) at each BS, a novel content placement approach is proposed in Section II-C of \cite{blaszczyszyn2015optimal}\footnote{Note that the file ordering in this content placement approach can be varied arbitrarily at each particular realization if different file combinations are desired.}. This approach brings dependency among the caching events of different files. But such cache dependency is irrelevant to the analysis in this work.
\subsection{Probability of Tier Association}
Without loss of generality, we carry out our analysis for a typical user, denoted as $u_0$, located at the origin as in \cite{andrews2011tractable}. In the cache-enable HetNet, the user association policy does not only depend on the received signal strength but also the requested and cached contents. Specifically, when $u_o$ requests content $j$, it is associated with the strongest BS among those that have cached content $j$ from all the tiers based on the average received signal power.
Denote the distance between $u_0$ and the nearest BS caching content $j$ in the $i$-th tier by $r_{i|j}$. According to our tier association policy, the index of the tier that $u_0$ is associated with for content $j$ is:
\begin{equation}
i(j)=\arg{\max_{l\in{\mathcal{N}}}\{B_{l}S_{l}r_{l|j}^{-\beta}\}},
\end{equation}
where $B_l$ and $S_l$ are the association bias factor and the transmit power of BSs in the $l$-th tier, respectively. For notation simplicity, we assume $B_l=1$, $\forall l\in\mathcal{N}$, in the rest of the paper.

It is essential to determine each tier's association probability when a user requests a content. Since each BS caches contents independently of other BSs, the locations of the BSs caching content $j$ in the $i$-th tier can be modeled as a thinned HPPP $\Phi_{ij}$ with density $\lambda_i{}p_{ij}$\cite{Martain}{\footnote{Note that application of the thinning property of HPPP is based on the assumption that each BS caches contents independently of other BSs, hence it is not affected by the realization method of \cite{blaszczyszyn2015optimal}.}}. Then, we have the following lemma.
\begin{lemma} \label{lemma1}
The probability of $u_0$ associated with the $i$-th tier for content $j$ is given by
\begin{align}
\mathcal{W}_{i|j}&=\frac{\lambda_{i}p_{ij}{S_i}^{\frac{2}{\beta}}}{\sum_{l=1}^{N}{\lambda_{l}p_{lj}{S_l}^{\frac{2}{\beta}}}}.  \label{Wij2}     \end{align}
\end{lemma}
\begin{proof}
Please refer to Appendix \ref{1}.
\end{proof}
This lemma states that the association probability $\mathcal{W}_{i|j}$ is determined directly by the density $\lambda_i p_{ij}$ and the transmit power $S_i$ of the thinned HPPP $\Phi_{ij}$.
\section{Performance Analysis} \label{performancemetric}
In this section, we analyze the SDP for a given probabilistic caching scheme $\{p_{ij}\}_{i\in\mathcal{N},j\in\mathcal{M}}$. 
Consider that all the BSs operate in the fully loaded state and share the common bandwidth \cite{andrews2011tractable}. By using an orthogonal multiple access strategy within a cell, the intra-cell interference is thus not considered here and only the interference introduced by inter-tier cells and intra-tier other cells is incorporated into analysis. Given that $u_0$ sends a request for content $j$ and is associated with the $i$-th tier, then the received instantaneous signal-to-interference-plus-noise ratio (SINR) of $u_0$ is given by
\begin{equation}
\text{SINR}_{i|j}=\frac{S_{i}|h_{i,o}|^2R_{i|j}^{-\beta}}{\sigma^2+\sum_{l=1}^{N}\sum_{k\in{\Phi_{l}}{\backslash}\{n_{io}\}}S_{l}|h_{l,k}|^2d_{l,k}^{-\beta}}, \label{SINR}
\end{equation}
where $R_{i|j}$ is the distance from $u_0$ to its serving BS $n_{io}$ in the $i$-tier tier, $d_{l,k}$ denotes the distance between $u_0$ and the $k$-th interfering BS in the $l$-th tier, $|h_{i,o}|^2(|h_{l,k}|^2)\sim \exp(1)$ is the small-scale fading channel gain between $u_0$ and the serving BS (the $k$-th interfering BS). The delivery of content $j$ from tier $i$ is successful when the received SINR of $u_0$ is larger than a
threshold $\tau$. Thus, the SDP of content $j$ from tier $i$ can be expressed as\footnote{Note that the SDP $\mathcal{C}_{i|j}$ in a cache-enabled network depends on both the average received signal strength and the caching distribution, which is different from the traditional coverage probability where $u_{o}$ is always associated with the strongest BS.}
\begin{equation}
\mathcal{C}_{i|j}\triangleq\mathbb{E}_{R_{i|j}}\left[\mathbb{P}\left[\text{SINR}_{i|j}>\tau|R_{i|j}\right]\right]. \label{coverageij}
\end{equation}
Recall that the locations of BSs caching content $j$ in the $i$-tier can be modeled as a thinned HPPP, then the probability density function (PDF) of $R_{i|j}$ is given below.\footnote{$\mathcal{W}_{i|j}$ and $f_{R_{i|j}}$ can also be obtained by applying the thinned HPPP $\Phi_{ij}$ to Lemma $1$ and Lemma $3$ of \cite{SINR}, respectively.}
\begin{lemma} \label{lemma2}
The PDF of $R_{i|j}$ is
\begin{equation}
f_{R_{i|j}}(r)=\frac{2\pi{p_{ij}}\lambda_{i}}{\mathcal{W}_{i|j}}e^{-\pi\sum_{l=1}^{N}{p_{lj}}\lambda_{l}\left(\frac{S_{l}}{S_{i}}\right)^{\frac{2}{\beta}}r^{2}}r.\label{fR}
\end{equation}
\end{lemma}
\begin{proof}
Please refer to Appendix \ref{Lemma2}.
\end{proof}
Note that different from the conventional network without caching where each user is associated with the strongest BS and there only exists one type of interfering BSs, in the multi-tier cache-enabled HetNet considered in this work, the interferences to $u_o$ when it is associated with the $i$-th tier for content $j$ can be divided into two groups. The first group of interferences come from all the BSs (except the serving BS $n_{io}$) in each tier that have stored content $j$, the locations of which can be modeled as a thinned HPPP, denoted as $\Phi_{\l j}$ with density $p_{lj}\lambda_l$ for $l\in \mathcal{N}$. The second group comes from all the BSs in each tier that do not cache content $j$, the locations of which can also be modeled as a thinned HPPP, denoted as $\Phi_{l j^-}$ with density $(1-p_{lj})\lambda_l$ for $l\in \mathcal{N}$. For the first interference group, the distance from $u_o$ to each interfering BS in $\Phi_{l j}$ is at least $\left(\frac{S_{l}}{S_{i}}\right)^{\frac{1}{\beta}}$ times of the distance from $u_o$ to its serving BS $n_{io}$ according to (\ref{Wij}) caused by our association policy. For the second group, the interfering BSs could be very close to $u_o$. By carefully handling these two groups of interferences in $N$-tier HetNets, we derive an analytical expression for $\mathcal{C}_{i|j}$ in the following proposition.
\begin{proposition} \label{proposition1}
The SDP $\mathcal{C}_{i|j}$ is
\begin{align}
&\mathcal{C}_{i|j}\!=\!\frac{2\pi{p_{ij}}\lambda_{i}}{\mathcal{W}_{i|j}}\!\!\int_{0}^{\infty}\!\!r\exp\left(\!-\frac{\tau{r}^{\beta}\sigma^2}{S_{i}}\!\right)\nonumber\\
&\exp\!\left(\!\!-\pi\!\sum_{l=1}^{N}\!\lambda_{l}\!\left(\!\frac{S_{l}}{S_{i}}\!\right)^{\frac{2}{\beta}}\!\!\!\left(p_{lj}H(\tau,\beta){\setlength\arraycolsep{0.5pt}+}(\!1{\setlength\arraycolsep{0.5pt}-}p_{lj}\!)D(\tau,\beta)
{\setlength\arraycolsep{0.5pt}+}\!p_{lj}\!\right)\!r^{2}\!\right)\!\mathrm{d}{r},   \label{Cij}
\end{align}
where
$H(\tau,\beta)\triangleq\frac{2\tau}{\beta-2}{}_2F_{1}(1,1-\frac{2}{\beta},2-\frac{2}{\beta},-\tau)$, and $D(\tau,\beta)\triangleq\frac{2}{\beta}\tau^{\frac{2}{\beta}}B(\frac{2}{\beta},1-\frac{2}{\beta})$. Furthermore,
${}_2F_{\cdot}(\centerdot)$ denotes the Gauss hypergeometric function, and $B(x,y)$ is the Beta function defined as $\int_{0}^{1}\mu^{x-1}(1-\mu)^{y-1}d{\mu}$.
\end{proposition}
\begin{proof}
Please refer to Appendix \ref{2}.
\end{proof}
By the law of total probability, the average SDP for $u_0$ is given by
\begin{equation}
\mathcal{C}\triangleq\sum_{j=1}^{M}\sum_{i=1}^{N}t_{j}\mathcal{W}_{i|j}\mathcal{C}_{i|j}. \label{C}
\end{equation}
Substituting (\ref{Cij}) and (\ref{Wij2}) into (\ref{C}), we obtain a tractable expression of $\mathcal{C}$ as follows
\begin{align}
&\mathcal{C}=\sum_{j=1}^{M}\sum_{i=1}^{N}2\pi{p_{ij}t_{j}}\lambda_{i}\int_{0}^{\infty}r\exp\left(-\frac{\tau{r}^{\beta}\sigma^2}{S_{i}}\right)\nonumber\\
&\exp\!\left(\!\!-\pi\!\sum_{l=1}^{N}\!\lambda_{l}\!\left(\!\frac{S_{l}}{S_{i}}\!\right)^{\frac{2}{\beta}}\!\!\!\left(p_{lj}H(\tau,\beta){\setlength\arraycolsep{0.5pt}+}(\!1{\setlength\arraycolsep{0.5pt}-}\!p_{lj})D(\tau,\beta)
{\setlength\arraycolsep{0.5pt}+}\!p_{lj}\!\right)\!r^{2}\!\right)\!\mathrm{d}{r}. \label{CCCC}
\end{align}
In the interference-limited scenario, where the noise power is very small compared with the interference power and
hence can be neglected, the expression (\ref{CCCC}) can be simplified.
\begin{corollary}
In the interference-limited scenario, i.e., $\sigma^2\to{0}$, the SDP $\mathcal{C}$ can be simplified as
\begin{align}
\mathcal{C^{'}}=\sum_{j=1}^{M}\frac{\sum_{i=1}^{N}\lambda_{i}S_{i}^{\frac{2}{\beta}}p_{ij}t_{j}}{\sum_{l=1}^{N}\lambda_{l}{S_{l}}^{\frac{2}{\beta}}\left[T(\tau,\beta)p_{lj}+D(\tau,\beta)\right]}, \label{C'2}
\end{align}
where $T(\tau,\beta)\triangleq{}H(\tau,\beta)-D(\tau,\beta)+1$.
\end{corollary}
\begin{proof}
Substituting $\sigma^2=0$ to (\ref{CCCC}), we then obtain (\ref{C'2}).
\end{proof}
Equation (\ref{CCCC}) and (\ref{C'2}) show a tractable expression and a closed-form expression for the SDP in the general regime and high SNR regime, respectively. The performance metric depends on four main factors: the number of tiers $N$, the caching probabilities $\{p_{ij}\}_{i\in\mathcal{N},j\in\mathcal{M}}$, the BS densities $\{\lambda_i\}_{i\in\mathcal{N}}$ and transmit powers $\{S_i\}_{i\in\mathcal{N}}$. In the rest of this paper, we shall focus on the interference-limited regime with high SNR.
\section{Caching Optimization and Analysis} \label{problem}
In this section, we formulate and solve the optimal caching problem for maximizing the SDP in the high SNR regime. Further, by considering the optimal caching problem in special cases, we establish an interesting connection between $N$-tier HetNets and single-tier networks.
\subsection{Caching Optimization for General Case}
The optimal caching problem of maximizing the SDP is formulated as
\begin{align}
\mathrm{P1}:\quad\max_{\mathbf{P}}&\quad \mathcal{C^{'}}(\mathbf{P}) \nonumber\\
\mathnormal{s.t.}&\quad (\ref{condition1}),(\ref{condition2}) \nonumber
\end{align}
\begin{proposition} \label{proposition2}
Problem $\mathrm{P1}$ is a concave optimization problem.
\end{proposition}
\begin{proof}
Please refer to Appendix \ref{3}.
\end{proof}
By Proposition \ref{proposition2}, we can use the standard interior point method to solve $\mathrm{P1}$ \cite{boyd2004convex}.
Let $\mathbf{P}^{*}=[p^{*}_{ij}]_{N\times{M}}$ denote the optimal solution of $\mathrm{P1}$. By the Karush-Kuhn-Tucker (KKT) conditions, the sufficient and necessary conditions for $\mathbf{P}^{*}$ can be stated in the following lemma.
\begin{lemma} \label{lemma3}
The optimal solution $\mathbf{P}^{*}$ of Problem $\mathrm{P1}$ satisfies the following sufficient and necessary conditions:
\begin{equation}
p^{*}_{ij}\!=\!
\min\!\left\{\left[\frac{1}{G_{i}}\!\left(\!\sqrt{\frac{V_{ij}E}{\eta_{i}}}\!-\!\!\sum_{k=1,\ne{i}}^{N}\!G_{k}p^{*}_{kj}\!-\!E\!\right)\!\right]^{+}, 1\right\},
 \label{p1*}
\end{equation}
for $\forall{i}\in{\mathcal{N}},\forall{j}\in{\mathcal{M}}$, where $[x]^{+}=\max\{0,x\}$, $V_{ij}=\lambda_{i}{S_{i}}^{\frac{2}{\beta}}t_{j}$, $G_{i}=\lambda_{i}{S_{i}}^{\frac{2}{\beta}}T(\tau,\beta)$,  $E=D(\tau,\beta)\sum_{i=1}^{N}\lambda_{i}{S_{i}}^
{\frac{2}{\beta}}$, and
$\eta_{i}$ is the Lagrangian multiplier that satisfies $\sum_{j=1}^{M}p^{*}_{ij}=Q_{i}$ \footnote{From (\ref{derivative}), it can be shown by contradiction that the maximum SDP is achieved only when constraint (\ref{condition2}) holds with equality. The similar proof is given by Lemma 2 of \cite{blaszczyszyn2015optimal}. In the rest of this paper, we use (\ref{condition2}) with equality as the constraint.} for $\forall{i}\in{\mathcal{N}}$.
\end{lemma}
\begin{proof}
Please refer to Appendix \ref{6}.
\end{proof}
Furthermore, according to (\ref{derivative}), we have the following remark to state the impact of the BS cache size of each tier.
\begin{remark} \label{remark_1}
For $\forall{i}\in\mathcal{N}$, the maximum SDP $\mathcal{C^{'*}}$ increases with the cache size $Q_{i}$ ($Q_{i}<M$).
\end{remark}
\subsection{Caching Optimization for Special Cases} \label{special}
\subsubsection{Optimization for N=1} \label{N=1}
When $N=1$, the network degrades to the single-tier network. Denote $\mathbf{P_{1}}\triangleq{(p_{j})_{1\times{M}}}$ as the caching strategy and $Q<{M}$ as the cache size for the single-tier network, then the optimal caching probabilities in (\ref{p1*}) become:
\begin{equation}
p_{j}^{*}\!=\!\min\!\left\{\left[\frac{1}{T(\tau,\beta)}\sqrt{\frac{t_{j}D(\tau,\beta)}{\eta^{*}}}{\setlength\arraycolsep{0.5pt}-}\frac{D(\tau,\beta)}{T(\tau,\beta)}\right]^{{\setlength\arraycolsep{0.5pt}+}},\!1\!\right\},~\forall{j}\in\!\mathcal{M} \label{pj*}
\end{equation}
where $\eta^{*}$ satisfies $\sum_{j=1}^{M}p_{j}^{*}=Q$ and can be found by bisection method. \footnote{Note that the optimal caching probability in this special case is consistent with the prior works on single-tier networks in \cite{7562510,cui2015analysis}. Our work extends the probabilistic caching strategy optimization for a single-tier network to that for a general $N$-tier HetNet and contributes to presenting the impacts and essential tradeoffs of the heterogeneous network parameters.}

Based on (\ref{pj*}), we have the following result.
\begin{corollary}\label{corollary22}
The optimal caching probability $p_{j}^{*}$ $(j\in\mathcal{M})$ decreases with the index $j$, i.e, $p_{j}^{*}$ increases with the content popularity $t_{j}$. Besides, $p_{j}^{*}$ $(j\in\mathcal{M})$ increases with the cache size $Q$.
\end{corollary}
\begin{proof}
Please refer to Appendix \ref{7}.
\end{proof}
By (\ref{C'2}) for $N=1$, the maximum SDP for the single-tier network, denoted by $\mathcal{C}_{1}^{'*}$, is \begin{equation}\mathcal{C}_{1}^{'*}=\sum_{j=1}^{M}\frac{{p^{*}_{j}}t_{j}}{T(\tau,\beta)p^{*}_{j}+D(\tau,\beta)},\label{1111}\end{equation} thus we have the following remark.
\begin{remark}\label{remark2}
In the interference-limited regime, the maximum SDP $\mathcal{C}^{'*}_{1}$ of single-tier networks is independent of the BS density, transmit power, and only depends on the cache size. This is because the serving BS and interfering BSs have the same caching resource, and the increase in signal power is counter-balanced by the increase in interference power. Similar performance independency on the BS density and transmit power also exists for traditional networks without cache \cite{andrews2011tractable}.
\end{remark}

\subsubsection{Connection between 1-tier and $N$-tier HetNets}
We observe from (\ref{C'2}) that the SDP in the high SNR regime depends on the caching probabilities $\{p_{ij}\}_{i\in{\mathcal{N}},j\in{\mathcal{M}}}$ through $\sum_{i=1}^{N}\lambda_{i}{S_{i}}^{\frac{2}{\beta}}p_{ij}$. Thus, by defining $x_{j}=\frac{\sum_{i=1}^{N}\lambda_{i}{S_{i}}^{\frac{2}{\beta}}p_{ij}}{\sum_{i=1}^{N}\lambda_{i}{S_{i}}^{\frac{2}{\beta}}}$, then we can formulate a new problem below,
\begin{align}
\mathrm{P2}:\quad\max_{\mathbf{x}}\quad &\sum_{j=1}^{M}\frac{{x_{j}}t_{j}}{T(\tau,\beta)x_{j}+D(\tau,\beta)} \label{C^{'}_{1}}\\
\mathnormal{s.t.}\quad &0\le{x_{j}}\le{1},~~\forall{j}\in{\mathcal{M}} \label{condition4}\\
&\sum_{j=1}^{M}{x_{j}}=\frac{\sum_{i=1}^{N}\lambda_{i}{S_{i}}^{\frac{2}{\beta}}Q_{i}}{\sum_{i=1}^{N}\lambda_{i}{S_{i}}^{\frac{2}{\beta}}}. \label{condition5}
\end{align}

From (\ref{1111}) and (\ref{C^{'}_{1}}), it is seen that $\mathrm{P2}$ is identical with the caching optimization problem in a single-tier network with $\mathbf{x}$ being the caching probability vector and $Q_{e}=\frac{\sum_{i=1}^{N}\lambda_{i}{S_{i}}^{\frac{2}{\beta}}Q_{i}}{\sum_{i=1}^{N}\lambda_{i}{S_{i}}^{\frac{2}{\beta}}}$ being the equivalent cache size. By further comparing $\mathrm{P1}$ and $\mathrm{P2}$, we obtain the following proposition which states a general relationship
between the optimal performance of an $N$-tier HetNet and that
of a single-tier network.
\begin{proposition} \label{proposition_3}
Let $\mathcal{C}^{'*}$ be the optimal objective of $\mathrm{P1}$ for the considered $N$-tier HetNet. For the single-tier network with cache size $Q_{e}=\frac{\sum_{i=1}^{N}\lambda_{i}{S_{i}}^{\frac{2}{\beta}}Q_{i}}{\sum_{i=1}^{N}\lambda_{i}{S_{i}}^{\frac{2}{\beta}}}$, we have
$$\mathcal{C}^{'*}\le\mathcal{C}^{'*}_{1}$$
with equality if $Q_i = Q_j$ for $\forall i,~j \in \mathcal{N}$.
\end{proposition}
\begin{proof}
Please refer to Appendix \ref{4}.
\end{proof}
Proposition \ref{proposition_3} states that in the interference-limited regime, the optimal performance of an $N$-tier HetNet with BS cache sizes $\{Q_{i}\}_{i\in\mathcal{N}}$, BS densities $\{\lambda_i\}_{i\in\mathcal{N}}$, and BS transmit powers $\{S_{i}\}_{i\in\mathcal{N}}$ is upper bounded by that of a single-tier network with BS cache size $Q_{e}=\frac{\sum_{i=1}^{N}\lambda_{i}{S_{i}}^{\frac{2}{\beta}}Q_{i}}{\sum_{i=1}^{N}\lambda_{i}{S_{i}}^{\frac{2}{\beta}}}$, arbitrary BS density, and arbitrary BS transmit power. Further, their performances are the same when the cache sizes $\{Q_i\}_{i\in\mathcal{N}}$ are the same for all tiers in the $N$-tier HetNet.
\subsubsection{Optimization for $Q_i = Q_j$, $\forall i,~j\in\mathcal{N}$}
In this case, Proposition \ref{proposition_3} shows that $\mathrm{P1}$ is equivalent to $\mathrm{P2}$.
Based on (\ref{pj*}), the optimal solution $\mathbf{P}^{*}$ of $\mathrm{P1}$ can be further given below.
\begin{corollary} \label{corollary3}
When $Q_i = Q_j$ for $\forall i,~j \in \mathcal{N}$, there exists an optimal solution $\mathbf{P}^{*}$ of $\mathrm{P1}$ satisfying
$p_{ij}^*=x_{j}^*$ for $\forall{i}\in{N}$, where the optimal solution $\mathbf{x}^*$ of $\mathrm{P2}$ follows (\ref{pj*}) with $p^*_j=x^*_j$ and $Q=Q_{e}$.
\end{corollary}
\begin{proof}
Problem $\mathrm{P2}$ is the new constructed single-tier caching optimization problem. Thus, its optimal solution $x^*_j$ is the same as (\ref{pj*}) where let $p^*_j=x^*_j$ and $Q=Q_{e}$. In the second part of the proof of Proposition \ref{proposition_3}, we show that when $Q_i = Q_j$ for $\forall i,~j \in \mathcal{N}$, $\mathrm{P1}$ is equivalent to $\mathrm{P2}$, and $\mathbf{P}^{*}$ with $p^*_{ij}=x^{*}_{j}$ is also an optimal solution of $\mathrm{P1}$. Thus, Corollary \ref{corollary3} is proved. 
\end{proof}
Based on Remark \ref{remark2} and Proposition \ref{proposition_3}, we have the following remark.
\begin{remark} \label{remark_3}
In the interference-limited regime, when all the BSs in the $N$-tier HetNet have the same cache size, the maximum SDP of the HetNet is independent of the network heterogeneity in the BS density and transmit power.
\end{remark}

Traditionally, without caching ability at BSs, the outage probability is independent of the number of tiers, the BS densities and transmit powers in the interference-limited $K$-tier HetNets \cite{SINR}. By introducing caching resource into the system, (\ref{C'2}) states that the maximum SDP $\mathcal{C^{'*}}$ generally depends on the number of tiers $N$, the BS cache size $\{Q_{i}\}_{i\in\mathcal{N}}$, the BS density $\{\lambda_{i}\}_{i\in\mathcal{N}}$ and transmit power $\{S_{i}\}_{i\in\mathcal{N}}$. The intuition behind this observation is that the caching resource changes the decision of a user to access a BS. In the multi-tier cache-enabled HetNet, the decision not only depends on the received SINR, but also depends on the contents cached at BSs.
In order to further understand the impact of the cache size $Q_{i}$ on $\mathcal{C^{'}}(\mathbf{P})$, we theoretically illustrate the relationship between the cache size and the BS density$/$transmit power in the next section.
\section{Analysis on Network Parameters under Uniform Cache} \label{tradeoffimpacts}
In this section, with the uniform caching strategy where each content is cached with equal probabilities regardless of content popularities, the equivalence between the SDP of an $N$-tier HetNet and that of a single-tier network is obtained. Based on this property, we further investigate the impacts of the key network parameters, i.e, the BS cache size, density and transmit power on the system performance. Finally, the tradeoffs of the BS density $\lambda_{i}$, transmit power $S_{i}$ and cache size $Q_{i}$ are found.
\subsection{The Equivalence under Uniform Cache}
Consider the uniform caching strategy where $\mathbf{P_{1,u}}=\{p_{j}=\frac{Q_{e}}{M}\}_{j\in{\mathcal{M}}}$ and $\mathbf{P_{u}}=\{p_{ij}=\frac{Q_{i}}{M}\}_{i\in{\mathcal{N}},j\in{\mathcal{M}}}$ for single-tier and $N$-tier HetNets, respectively.
By substituting $\mathbf{P_{u}}$ and $\mathbf{P_{1,u}}$ into (\ref{C'2}) and (\ref{C^{'}_{1}}), respectively, the equivalence of the system performance of $N$-tier HetNets and 1-tier Networks can be established.
\begin{proposition}\label{proposition_4}
In the interference-limited regime, the SDP $\mathcal{C}^{'}(\mathbf{P_{u}})$ of a $N$-tier HetNet with the caching strategy $\mathbf{P_{u}}$ equals that $\mathcal{C}^{'}_{1}(\mathbf{P_{1,u}})$ of the single-tier network with the caching strategy $\mathbf{P_{1,u}}$, and is given by
\begin{equation}\mathcal{C}^{'}(\mathbf{P_{u}})=\mathcal{C}^{'}_{1}(\mathbf{P_{1,u}})=\sum_{j=1}^{M}\frac{{Q_{e}}t_{j}}{T(\tau,\beta)Q_{e}+D(\tau,\beta)M},\label{C^{'}_{1e}}\end{equation}
where the cache size of the single-tier network $Q_{e}=\frac{\sum_{i=1}^{N}\lambda_{i}{S_{i}}^{\frac{2}{\beta}}Q_{i}}{\sum_{i=1}^{N}\lambda_{i}{S_{i}}^{\frac{2}{\beta}}}$.
\end{proposition}
\begin{proof}
Substituting $\mathbf{P_{u}}$ and $\mathbf{P_{1,u}}$ into (\ref{C'2}) and (\ref{C^{'}_{1}}), respectively, we then have (\ref{C^{'}_{1e}}).
\end{proof}
Based on (\ref{pj*}) and $\mathcal{C}^{'}(\mathbf{P_{u}})=\mathcal{C}^{'}_{1}(\mathbf{P_{1,u}})$, we have the following corollary.
\begin{corollary}
In the interference-limited regime, consider the scenario that the video content popularity distribution is an uniform distribution, i.e., $t_{j}=\frac{1}{M}$, $\forall{}j\in\mathcal{M}$, then the uniform caching strategies $\mathbf{P_{1,u}}$ and $\mathbf{P_{u}}$ are the optimal probabilistic caching strategies for single-tier networks and $N$-tier HetNets, respectively.
\end{corollary}
\begin{proof}
When $t_{j}=\frac{1}{M}$, $\forall{}j\in\mathcal{M}$, we have $p^{*}_{1}=p^{*}_{2}=\cdots=p^{*}_{M}$ from (\ref{pj*}). Since $\sum_{j=1}^{M}p_{j}^{*}=Q_{e}$, we thus have $p^{*}_{j}=\frac{Q_e}{M}$ for $\forall{}j\in\mathcal{M}$, and the maximum $\mathcal{C}^{'*}_{1}=\mathcal{C}^{'}_{1}(\mathbf{P_{u}})$. Then from Proposition \ref{proposition_3} and Proposition \ref{proposition_4}, we have $\mathcal{C}^{'*}=\mathcal{C}^{'}(\mathbf{P_{u}})$, and hence $\mathbf{P_{u}}$ is also the optimal solution of $\mathrm{P1}$.
\end{proof}
\begin{figure*}
\vspace{-5mm}
\centering
\subfigure[The tradeoff between $\lambda_{1}$ and $Q_{1}$. We consider two cases: (1). \{$S_{1}$, $S_{2}$\}=\{43 dBm, 33 dBm\} (2). \{$S_{1}$, $S_{2}$\}=\{53 dBm, 33 dBm\}. We fix $\lambda_{2}$=$\frac{5}{\pi{500}^{2}}$ for both cases.]
{\label{trade1} 
 \includegraphics[width=3.1in, height=2.3in]{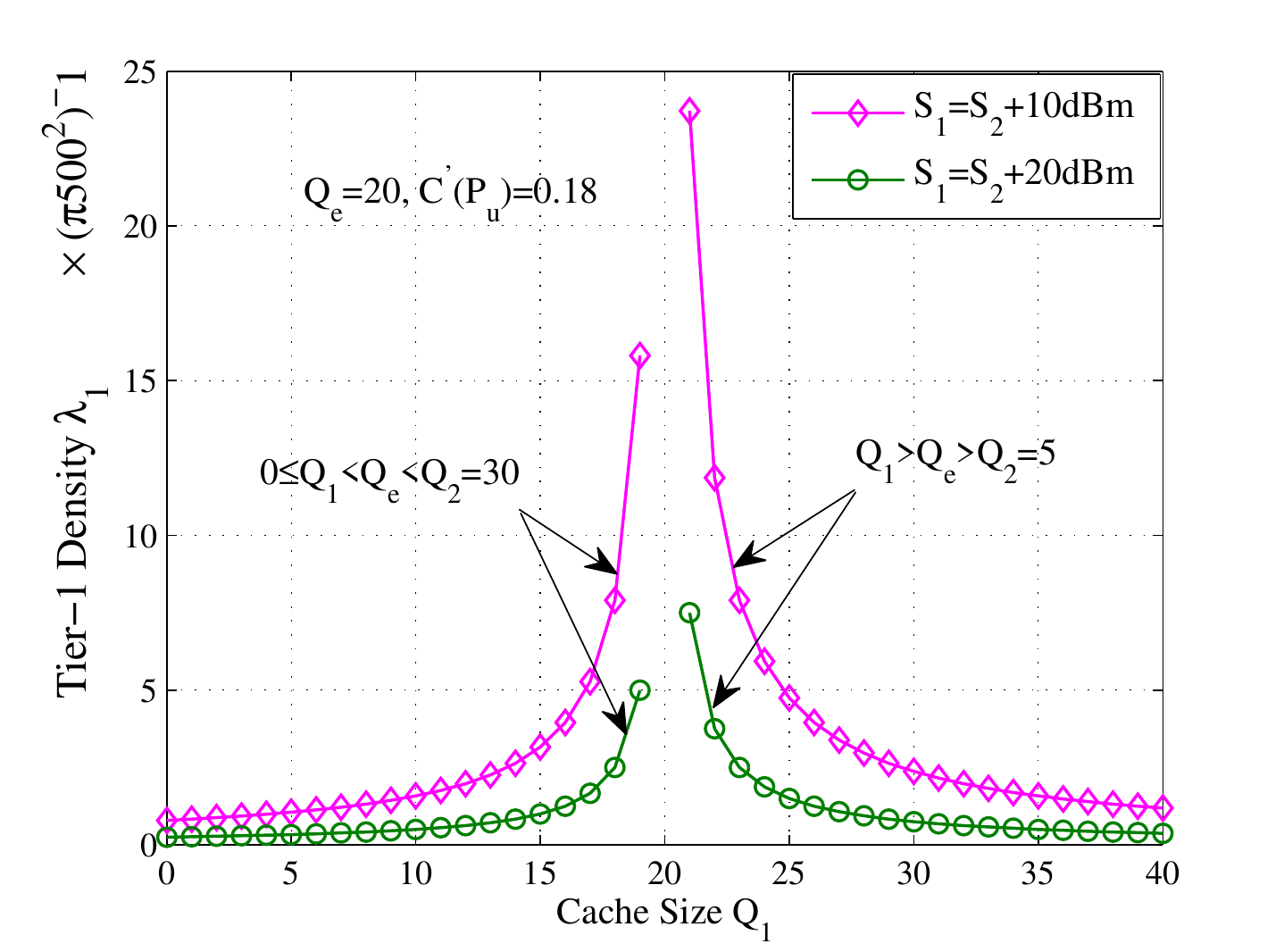}}
\hspace{0in}
\subfigure[The tradeoff between $S_{1}$ and $Q_{1}$. We consider two cases: (1). \{$\lambda_{1}$, $\lambda_{2}$\}=$\left\{\frac{2}{\pi{500}^{2}},~\frac{10}{\pi{500}^{2}}\right\}$ (2). \{$\lambda_{1}$, $\lambda_{2}$\}=$\left\{\frac{1}{\pi{500}^{2}},~\frac{10}{\pi{500}^{2}}\right\}$. We fix $S_{2}$=33 dBm for both cases.]
{\label{trade2} 
 \includegraphics[width=3.1in, height=2.3in]{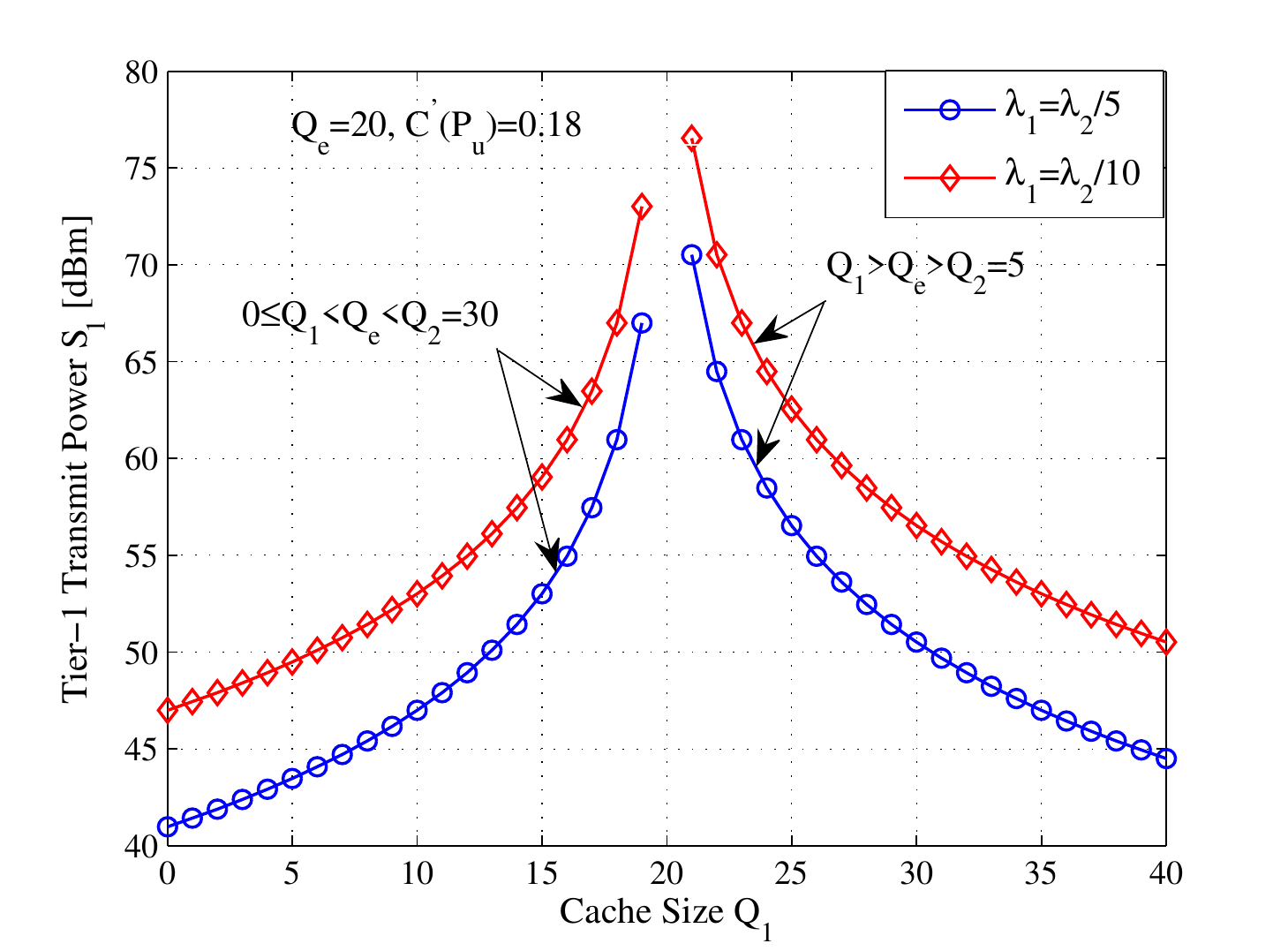}}
 \vspace{-1mm}
 \caption{Under fixed $\mathcal{C}^{'}(\mathbf{P_{u}})$, the tradeoffs of the BS density $\lambda_{1}$, transmit power $S_{1}$ and cache size $Q_{1}$ in a two-tier HetNet where $\tau=-10$ dB, $t_{j}={\frac{1/j^{\gamma}}{\sum_{k=1}^{M}1/k^{\gamma}}}$, $\gamma=0.8$, $\beta=4$, $\mathcal{C}^{'}(\mathbf{P_{u}})=0.18$ ($Q_{e}$=20). By Corollary \ref{corollary7}, we have $K_{1}<0$, $K_{2}<0$ for $Q_{1}\in[0,Q_{e})$, and $K_{1}>0$, $K_{2}>0$ for $Q_{1}\in(Q_{e},+\infty)$.}
\label{tradeoff_1_2} 
\vspace{-5mm}
\end{figure*}
\subsection{Impacts of $\lambda_{i}$, $S_{i}$, $Q_{i}$, $\forall{i\in{\mathcal{N}}}$}
Proposition \ref{proposition_4} further states that, with uniform caching regardless of content popularity, the SDP performance of $N$-tier HetNets depends on all the system parameters through the equivalent cache size $Q_{e}=\frac{\sum_{i=1}^{N}\lambda_{i}{S_{i}}^{\frac{2}{\beta}}Q_{i}}{\sum_{i=1}^{N}\lambda_{i}{S_{i}}^{\frac{2}{\beta}}}$ only. Then, the impacts of the BS density $\lambda_{i}$ and cache size $Q_{i}$ of tier $i$ can be obtained in the following two lemmas.
\begin{lemma} \label{lemma4}
The SDP $\mathcal{C}^{'}(\mathbf{P_{u}})$ increases with $S_{i}$ and $\lambda_{i}$ when $Q_{i}\ge{\frac{\sum_{j=1,\ne{i}}^{N}\lambda_{j}{S_{j}}^{\frac{2}{\beta}}Q_{j}}{\sum_{j=1,\ne{i}}^{N}\lambda_{j}{S_{j}}^{\frac{2}{\beta}}}}$. Otherwise, $\mathcal{C}^{'}(\mathbf{P_{u}})$ decreases with $S_{i}$ or $\lambda_{i}$.
\end{lemma}
\begin{proof}
From (\ref{C^{'}_{1e}}), we can see that $\mathcal{C}^{'}_{1}(\mathbf{P_{1,u}})$ increases with $Q_{e}$. Then we have
\begin{align}
\frac{\partial{Q_e}}{\partial{\lambda_{i}}}&=\frac{S^{\frac{2}{\beta}}_{i}
(Q_{i}\sum_{j=1}^{N}\lambda_{j}{S_{j}}^{\frac{2}{\beta}}-\sum_{j=1}^{N}\lambda_{j}
{S_{j}}^{\frac{2}{\beta}}Q_{j})}{(\sum_{j=1}^{N}\lambda_{j}{S_{j}}^{\frac{2}{\beta}})^2}\nonumber\\
&=\frac{S^{\frac{2}{\beta}}_{i}
(Q_{i}\sum_{j=1,\ne{i}}^{N}\lambda_{j}{S_{j}}^{\frac{2}{\beta}}-\sum_{j=1,\ne{i}}^{N}\lambda_{j}
{S_{j}}^{\frac{2}{\beta}}Q_{j})}{(\sum_{j=1}^{N}\lambda_{j}{S_{j}}^{\frac{2}{\beta}})^2}\nonumber,\\
\frac{\partial{Q_e}}{\partial{S_{i}}}&=\frac{\frac{2}{\beta}\lambda_{i}S^{\frac{2}{\beta}-1}_{i}(Q_{i}\sum_{j=1,\ne{i}}^{N}\lambda_{j}{S_{j}}^
{\frac{2}{\beta}}\!-\!\sum_{j=1,\ne{i}}^{N}\lambda_{j}{S_{j}}^{\frac{2}{\beta}}Q_{j})}{(\sum_{j=1}^{N}\lambda_{j}{S_{j}}^{\frac{2}{\beta}})^2} \nonumber.
\end{align}
Obviously, when $Q_{i}\ge{\frac{\sum_{j=1,\ne{i}}^{N}\lambda_{j}{S_{j}}^{\frac{2}{\beta}}Q_{j}}{\sum_{j=1,\ne{i}}^{N}\lambda_{j}{S_{j}}^{\frac{2}{\beta}}}}$, we have that $\frac{\partial{Q}}{\partial{\lambda_{i}}}\ge{0}$ and $\frac{\partial{Q}}{\partial{S_{i}}}\ge{0}$, which also means that $\mathcal{C}^{'}_{1}(\mathbf{P_{1,u}})$ increases with $\lambda_{i}$ and ${S_{i}}$ \footnote{This is not contradict with Remark \ref{remark2} in Section \ref{problem}, because the BS densities $\{\lambda_{i}\}_{i\in\mathcal{N}}$, transmit powers $\{S_{i}\}_{i\in\mathcal{N}}$ and cache sizes $\{Q_{i}\}_{i\in\mathcal{N}}$ affect the cache size $Q$ in here based on $Q_{e}$.}. Due to $\mathcal{C}^{'}(\mathbf{P_{u}})$=$\mathcal{C}^{'}_{1}(\mathbf{P_{1,u}})$, we thus have Lemma \ref{lemma4}.
\end{proof}
\begin{lemma} \label{lemma5}
The SDP $\mathcal{C}^{'}(\mathbf{P_{u}})$ increases with $Q_{i}$ ($Q_{i}<M$) and the increasing speed is monotonic to $\lambda_{i}{S_{i}}^{\frac{2}{\beta}}$, $\forall{i}\in{\mathcal{N}}$.
\end{lemma}
\begin{proof}
Due to $
\frac{\partial{Q_{e}}}{\partial{Q_{i}}}=\frac{\lambda_{i}{S_{i}}^{\frac{2}{\beta}}}{(\sum_{i=1}^{N}\lambda_{i}{S_{i}}^{\frac{2}{\beta}})^2} >0\nonumber
$, $Q_{e}$ increases with $Q_{i}$ for $\forall{i}\in{N}$. Since $\mathcal{C}^{'}_{1}(\mathbf{P_{1,u}})$ increases with $Q_{e}$ and $\mathcal{C}^{'}_{1}(\mathbf{P_{1,u}})=\mathcal{C}^{'}(\mathbf{P_{u}})$, $\mathcal{C}^{'}(\mathbf{P_{u}})$ increases with $Q_{i}$. We thus have this lemma.
\end{proof}
\begin{remark}
From Lemma \ref{lemma4} and Lemma \ref{lemma5}, we can observe two important features in the $N$-tier cache-enabled HetNet:
\begin{itemize}
  \item Increasing the density or transmit power of the BSs from the tier with small cache size decreases the system performance. This somewhat surprising result is actually intuitive since such BSs (e.g., pico or femto) with small cache only provide little service but bring strong interferences to other BSs (e.g., macro or relay).
  \item  If the BS transmit power and density of the $i$-th tier are both the largest among all the tiers, it is most effective to increase the performance of the network by increasing the BS cache size of the $i$-th tier.
\end{itemize}
\end{remark}
\subsection{Tradeoffs of $Q_{i}$, $\lambda_{i}$, and $S_{i}$, $\forall{i}\in{\mathcal{N}}$}
In the preceding subsection, we describe the impacts of $Q_{i}$, $\lambda_{i}$, and $S_{i}$ on $\mathcal{C}^{'}(\mathbf{P_{u}})$. Now we will present the tradeoffs of these network parameters at a target SDP.
\subsubsection{Tradeoffs of one tier parameters}
Lemma \ref{lemma4} and \ref{lemma5} show that increasing the BS density, transmit power or cache size influence the SDP $\mathcal{C}^{'}(\mathbf{P_{u}})$ by changing $Q_{e}$. This suggests that, as long as $Q_{e}$ does not change, one can interchange different types of system parameters to maintain the same system performance, and hence can obtain the tradeoffs of these network parameters at a target SDP. In the following, we elaborate the tradeoff between the BS density $\lambda_i$ (transmit power $S_{i}$) and the cache size $Q_i$ within each tier $i$ for $\forall i\in\mathcal{N}$.

Given a target SDP $\mathcal{C}^{'}(\mathbf{P_{u}})$,
the communication resource ($\lambda_{j}$ and $S_{j}$) and the caching resource ($Q_{j}$) of all the tiers without the $i$-th tier, i.e., $\forall{j}\in\mathcal{N}$ and $j\neq i$, we can obtain $Q_{e}$ based on (\ref{C^{'}_{1e}}), thereby gaining the tradeoffs of the $i$-th tier's cache size $Q_{i}$, BS density $\lambda_{i}$ and transmit power $S_{i}$, as illustrated in the following theorem.
\begin{theorem} \label{proposition5}
 With the uniform caching strategy, given a target $\mathcal{C}^{'}(\mathbf{P_{u}})$ determined by $\mathcal{C}^{'}_{1}(\mathbf{P_{1,u}})$, and the fixed values $\lambda_{j}$, $S_{j}$, $Q_{j}$ for $\forall{j}\in\mathcal{N}$ and $j\neq i$, the network parameters $Q_{i}$, $\lambda_{i}$, $S_{i}$ satisfy the following tradeoffs
\begin{align}
\lambda_{i}&=\frac{K_{1}}{Q_{i}-Q_{e}},~\text{for given} ~S_{i},~\forall{i}\in{\mathcal{N}},\label{tradeoff1}\\
S_{i}&=(\frac{K_{2}}{Q_{i}-Q_{e}})^{\frac{\beta}{2}},~\text{for given} ~\lambda_{i}, \forall{i}\in{\mathcal{N}},\label{tradeoff2}
\end{align}
where
\begin{align}
 K_{1}&=\sum_{j=1,\ne{i}}^{N}{\lambda_{j}{(\frac{S_{j}}{S_{i}})}^{\frac{2}{\beta}}(Q_{e}-Q_{j})}
, \label{K_{2}}\\ K_{2}&=\sum_{j=1,\ne{i}}^{N}(\frac{\lambda_{j}}{\lambda_{i}}){S_{j}}^{\frac{2}{\beta}}(Q_{e}-Q_{j})
.\label{K_{33}}
\end{align}
\end{theorem}
\begin{proof}
Please refer to Appendix \ref{proofpro5}.
\end{proof}
\textbf{Interestingly, we can observe from Theorem \ref{proposition5} that $\lambda_{i}$ is inversely proportional to $Q_{i}$, while $S_{i}$ is a power function of $Q_{i}$ with a negative exponent ($-\beta/2$)}. Accordingly, it is natural to ask that if the BS density can be reduced by increasing the BS caching capability. If yes, what is the condition? Thus, we have the following corollary to answer this question.
\begin{corollary} \label{corollary7}
To maintain the same SDP, we have the following results for the $i$-th tier
\begin{itemize}
  \item The BS density $\lambda_{i}$ and transmit power $S_{i}$ decrease with the cache size $Q_{i}$, i.e., $K_{1}>0$ and $K_{2}>0$,, when $Q_{i}\in(Q_{e},\infty)$.
  \item The BS density $\lambda_{i}$ and transmit power $S_{i}$ increase with the cache size $Q_{i}$, i.e., $K_{1}<0$ and $K_{2}<0$, when $Q_{i}\in[0, Q_{e})$.
\end{itemize}

\end{corollary}
\begin{proof}
Please refer to Appendix \ref{proofcorollary7}.
\end{proof}
It is worth mentioning that \emph{increasing the BS caching capability of one tier does not always reduce the BS density or transmit power of this tier to achieve the same performance in a multi-tier NetHet}, in contrast to single-tier caching networks \cite{bacstug2015cache}. Taking $N=2$ for example, the tradeoff between $\lambda_{1}$ and $Q_{1}$ is shown in Fig. \ref{trade1}. The target SDP determined by $\mathcal{C}^{'}_{1}(\mathbf{P_{1,u}})$ is $0.18$. Note that if a tier is deployed with larger cache size ($Q_{1}>Q_{e}=20$), the more the BS cache size ($Q_{1}$) increases, the less the BS density ($\lambda_{1}$) becomes, in order to maintain the same $\mathcal{C}^{'}(\mathbf{P_{u}})$. However, if the tier is deployed with small cache size ($Q_{1}<Q_{e}=20$), increasing both the BS cache size and density can keep the same performance, as show in Fig. \ref{trade1}. This is because the increase in BS density of the tier with small cache size only improves its own tier association probability, and also causes the stronger interference to other tiers with larger cache sizes. Therefore, more users are associated with the tier with small cache size, and the SDP decreases because this tier can only serve fewer content requests. 
The tradeoff between $S_{1}$ and $Q_{1}$ shown in Fig. \ref{trade2} is similar to that of $\lambda_{1}$ and $Q_{1}$.
\subsubsection{Tradeoffs of different tier parameters}
Similarly, we can find the tradeoffs of the parameters in different tiers, as specified in Theorem \ref{proposition6}.
\begin{theorem}\label{proposition6}
 When $Q_{i}$, $\lambda_{j}$, and $S_{j}$ ($i,~j\in{\mathcal{N}},~i\ne{j}$) satisfy the following tradeoffs for any two different tiers, $\mathcal{C}^{'}(\mathbf{P_{u}})$ remains constant.
\begin{align}
\lambda_{j}&=\frac{K_{3}-K_{4}Q_{i}}{K_{5}}, \label{tradeoff3} \\
S_{j}&=(\frac{K_{3}-K_{4}Q_{i}}{K_{6}})^{\frac{\beta}{2}}, \label{tradeoff4}
\end{align}
for $\forall{i,j}\in{\mathcal{N}},~i\ne{j}$, where
\begin{align}
K_{3}&=Q_{e}\sum_{k=1,\ne{j}}^{N}\lambda_{k}{S_{k}}^{\frac{2}{\beta}}-{\sum_{k=1,\ne{i,j}}^{N}\lambda_{k}{S_{k}}^{\frac{2}{\beta}}Q_{k}},\label{K_{3}}\\
K_{4}&=\lambda_{i}S_{i}^{\frac{2}{\beta}}, \label{K_{4}}\\
K_{5}&=S_{j}^{\frac{2}{\beta}}(Q_{j}-Q_{e}),\label{K_{5}}\\
K_{6}&=\lambda_{j}(Q_{j}-Q_{e}).~~~~~~~~~~~~~~~~~~~~~\label{K_{6}}
\end{align}
\end{theorem}
\begin{proof}
When $\mathcal{C}_{1}^{'}(\mathbf{P_{1,u}})$ keeps constant, $Q_e$ will be a fixed value. We consider the tradeoffs between the $j$-th tier parameters $\lambda_{j}$, $S_{j}$ and the $i$-th tier parameter $Q_{i}$ assuming that other parameters are constants. Then according to $Q_{e}=\frac{\sum_{i=1}^{N}\lambda_{i}{S_{i}}^{\frac{2}{\beta}}Q_{i}}{\sum_{i=1}^{N}\lambda_{i}{S_{i}}^{\frac{2}{\beta}}}$, we have
\begin{align}
Q_{e}\!\!\!\sum_{k=1,\ne{j}}^{N}\!\!\!\lambda_{k}{S_{k}}^{\frac{\!2\!}{\!\beta\!}\!}-\!\!{\!\!\!\sum_{k=1,\ne{i,j}}^{N}\!\!\!\lambda_{k}{S_{k}}^
{\frac{\!2\!}{\!\beta}\!}Q_{k}}\!=\!\lambda_{i}S_{i}^{\frac{2}{\beta}}\!Q_{i}\!+\!\lambda_{j}{S_{j}}^{\frac{2}{\beta}}\!(Q_{j}\!-\!Q_{e}),\label{T1}
\end{align}
Substituting (\ref{K_{3}}), (\ref{K_{4}}), (\ref{K_{5}}), (\ref{K_{6}}) into (\ref{T1}), we thus have (\ref{tradeoff3}), (\ref{tradeoff4}).
\end{proof}
\begin{figure}
 \vspace{-2.8mm}
\begin{minipage}[t]{0.5\textwidth}
\centering
\includegraphics[width=3.0in, height=2.2in]{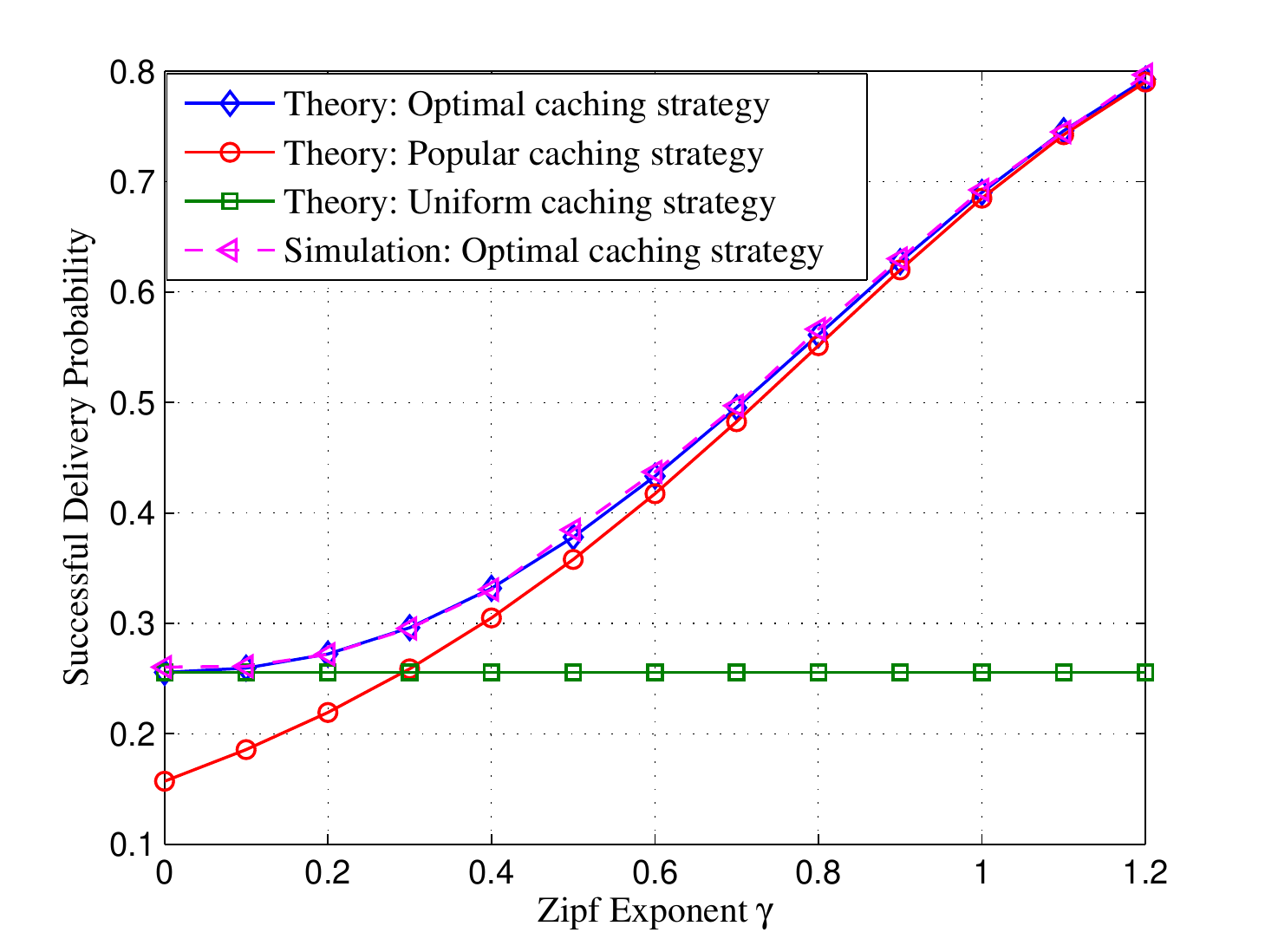}
\vspace{-4.5mm}
\caption{SDP with different caching strategies. Here, we set $M=1000$ and $\{Q_{1},Q_{2}\}$=$\{200, 50\}$.}
\label{3compare}
 \end{minipage}
 \hspace{0.05in}
 \begin{minipage}[t]{0.5\textwidth}
 \centering
\includegraphics[width=3.0in, height=2.2in]{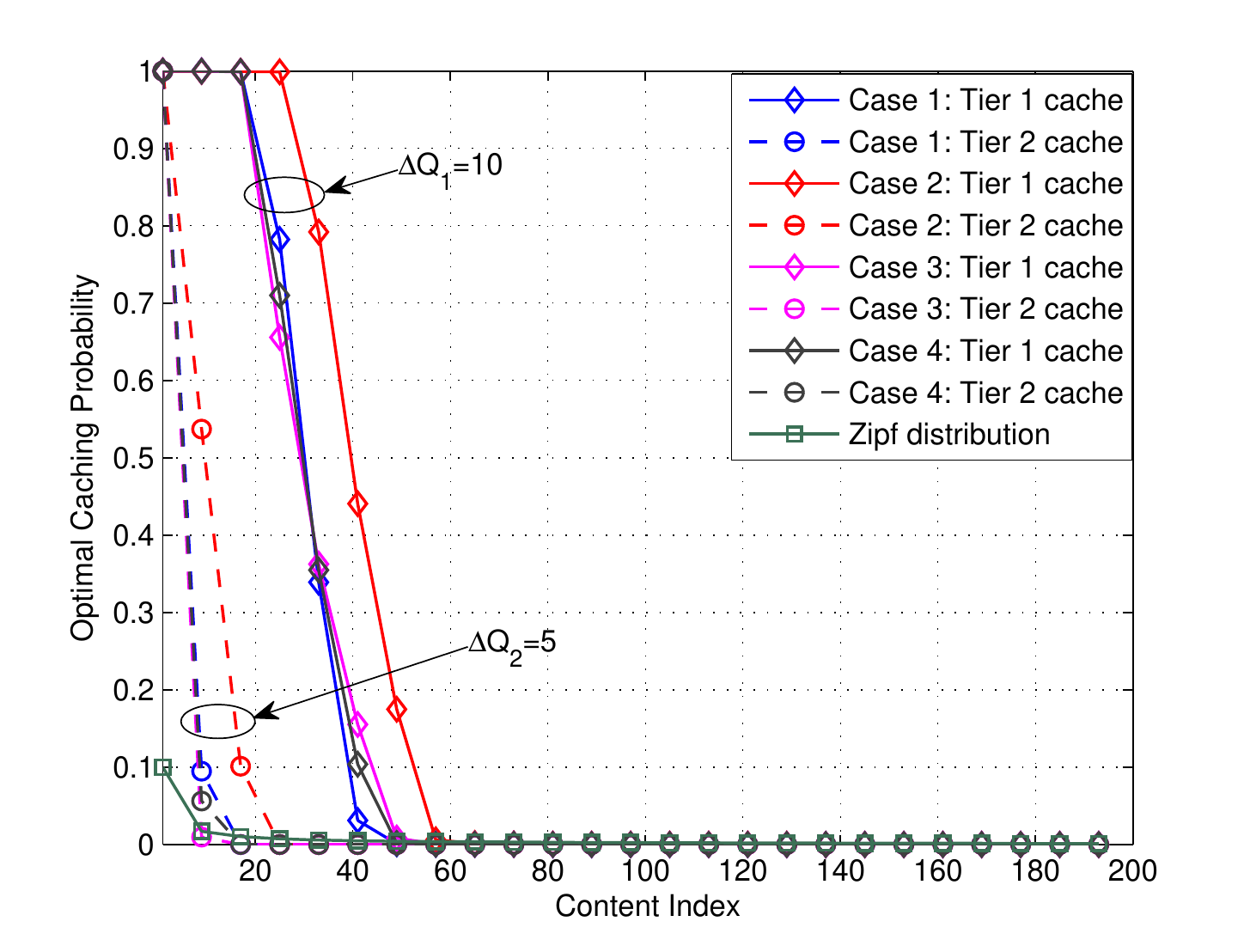}
\vspace{-4.5mm}
\caption{The optimal caching probability $p_{ij}^{*}$ for different storage sizes and transmit powers.}
\label{probability}
 \end{minipage}
  \vspace{-6mm}
\end{figure}
\begin{figure*}
\centering
\subfigure[Cache size $Q_{1}$. Here, we use $\{\lambda_{1}, \lambda_{2}\}=\{\frac{1}{\pi{500^2}}$, $\frac{5}{\pi{500^2}}$\} or \{$\frac{1}{\pi{500^2}}$, $\frac{10}{\pi{500^2}}\}$, $\{S_{1}$, $S_{2}\}=\{$43 dBm, 33 dBm\} or \{53 dBm, 33 dBm$\}$.]{
\label{mathbf{Q}}
  \includegraphics[width=3.2in, height=2.4in]{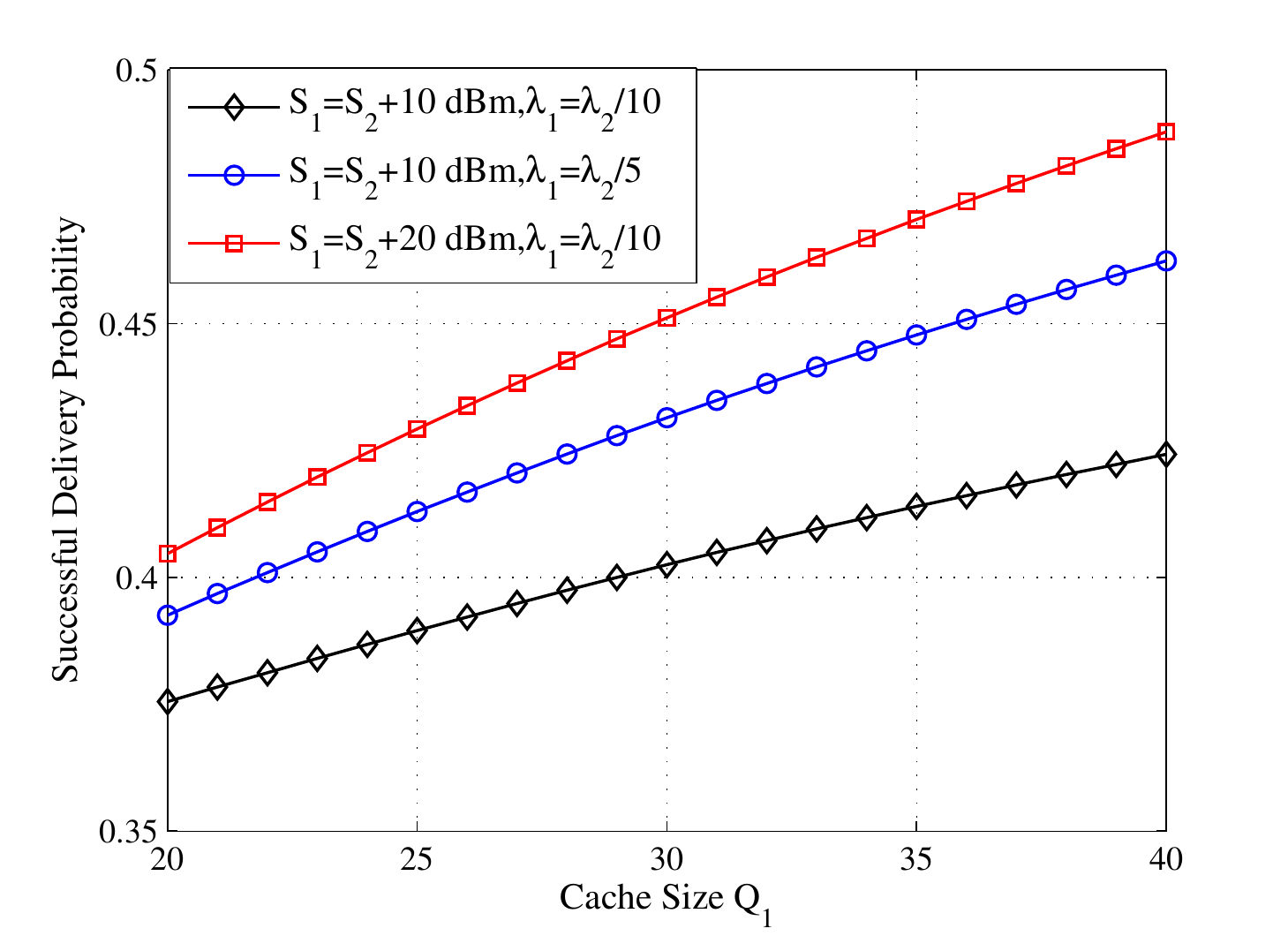}}
  \vspace{-2mm}
 \hspace{-0.2in}
 \subfigure[Densities $\{\lambda_{1}$, $\lambda_{2}\}$.]{
\label{lambda_{1}}
  \includegraphics[width=3.2in, height=2.4in]{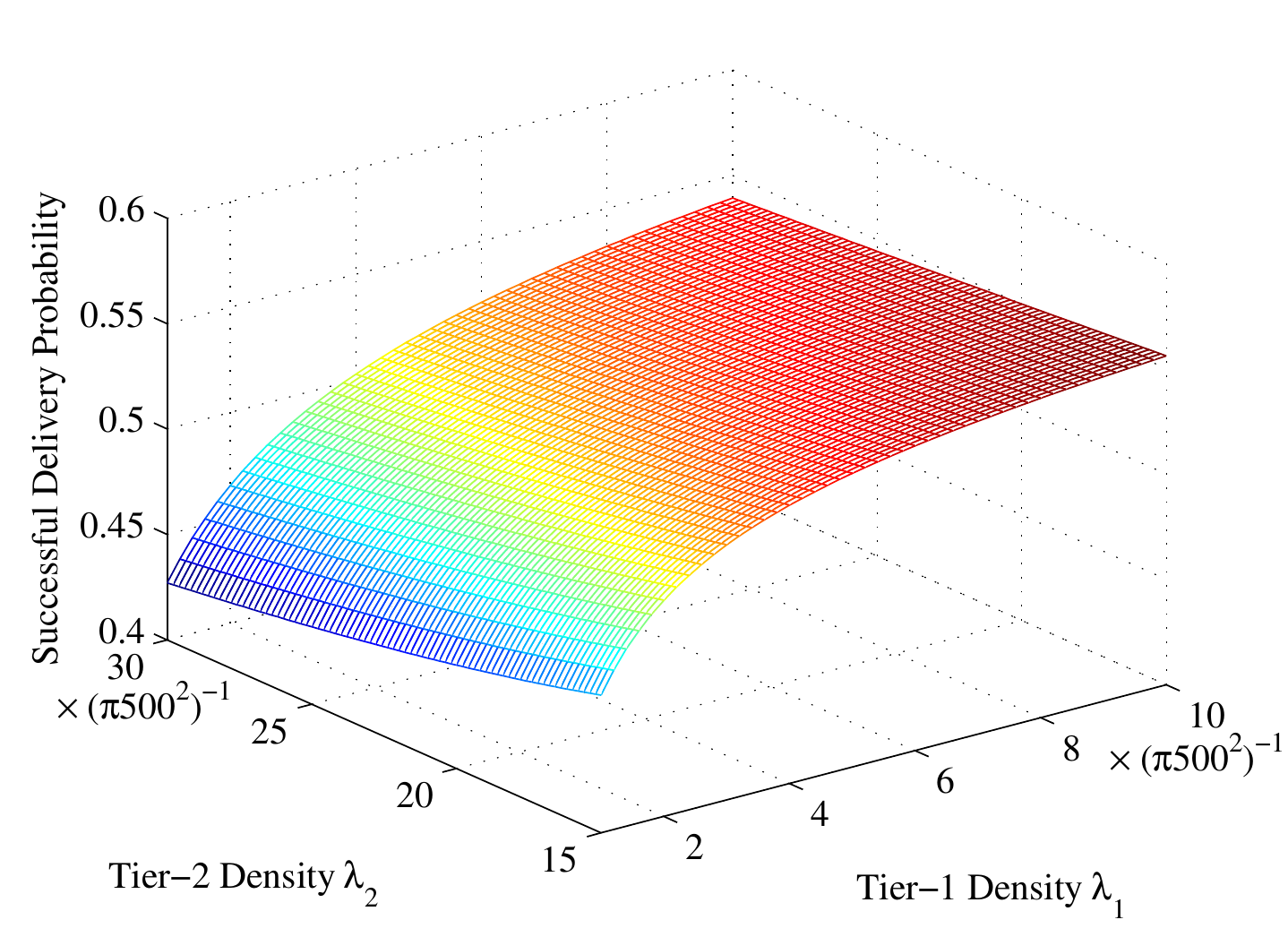}}
  \vspace{-2mm}
 \hspace{-0.1in}
 \subfigure[Tier-1 transmit power $S_{1}$. Here, we set $M=1000$.]{
\label{S_{1}} 
  \includegraphics[width=3.2in, height=2.4in]{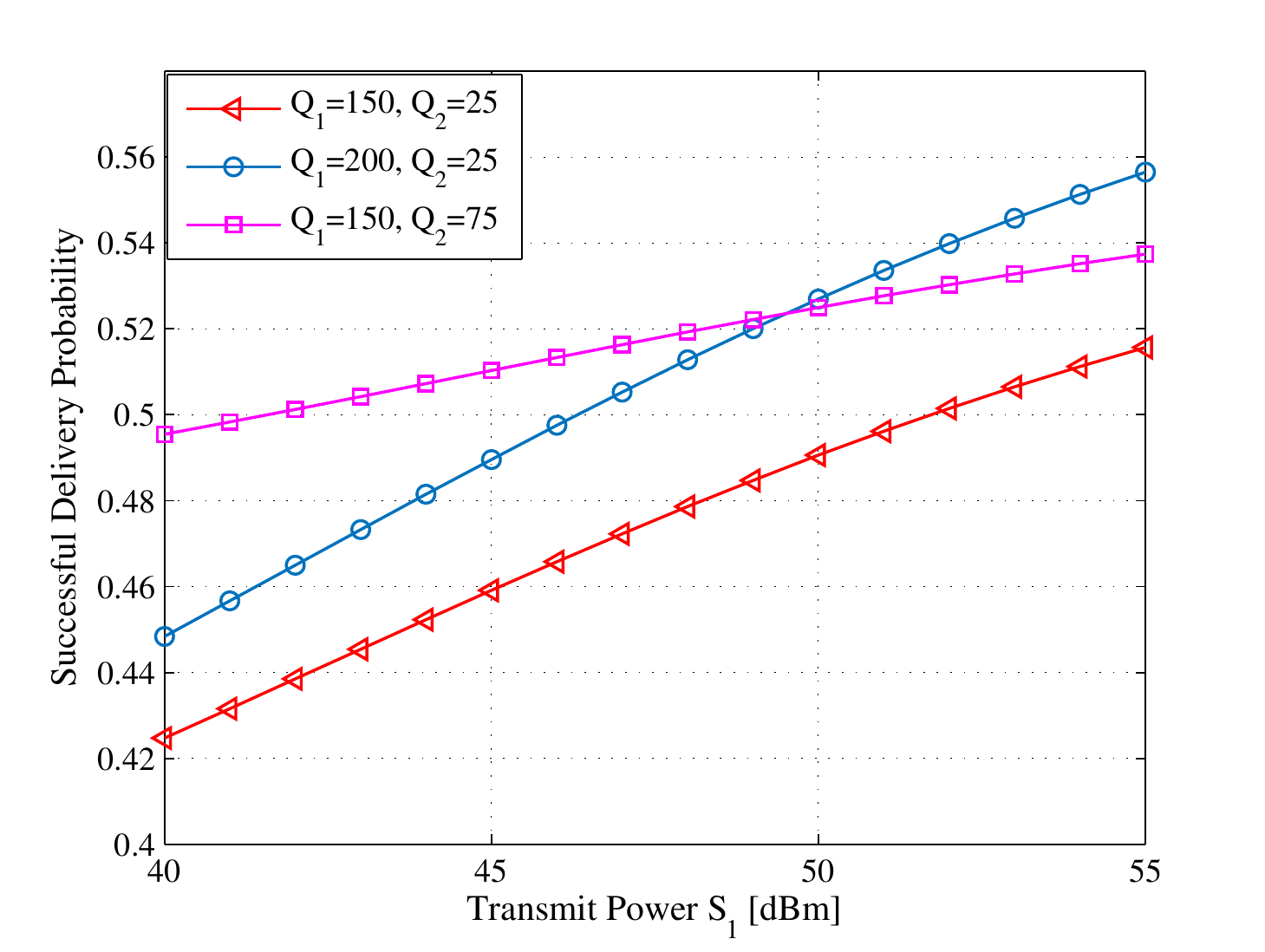}}
 \hspace{-0.2in}
 \subfigure[Tier-2 transmit power $S_{2}$.]{
   \label{S_{2}} 
    \includegraphics[width=3.2in, height=2.4in]{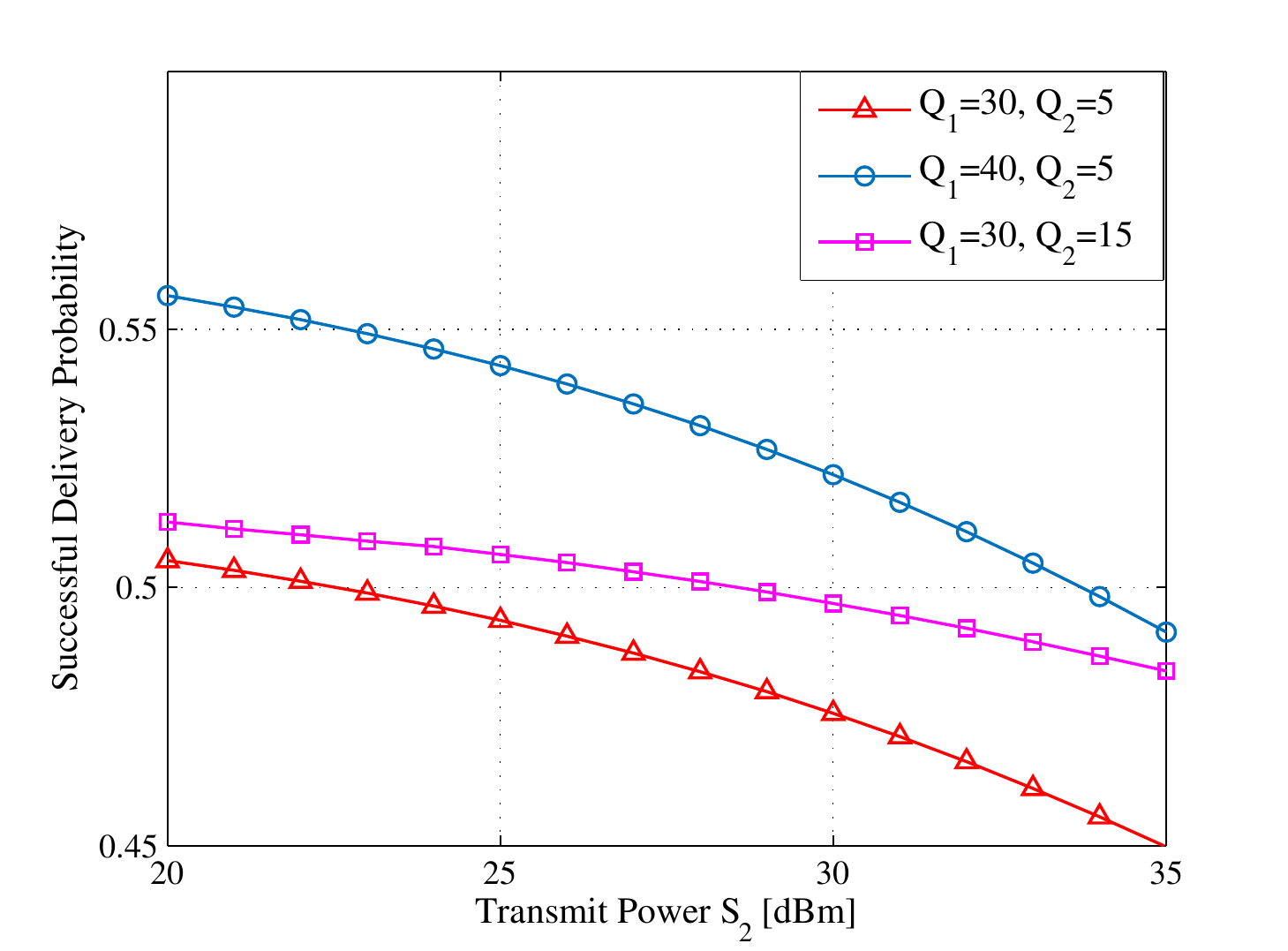}}
    \vspace{-1mm}
\caption{The impacts of the BS cache size $Q_{1}$, densities $\{\lambda_{1}$, $\lambda_{2}\}$, transmit powers $\{S_{1}$, $S_{2}\}$ on the maximum SDP $\mathcal{C}^{'*}$.}
\vspace{-4mm}
\end{figure*}
Different than Theorem \ref{proposition5}, \textbf{the BS density $\lambda_{j}$ is a linear function of the BS cache size $Q_{i}$, and the BS transmit power $S_{j}$ is a power function of $Q_{i}$ with a positive exponent $\frac{\beta}{2}$, where $i\ne j$}. Then, the function changing law follows the following corollary.
\begin{corollary} \label{coro8}
For different tiers, we have
\begin{itemize}
  \item When $Q_{j}>Q_{e}$, we have $K_{5}>0$, $K_{6}>0$ and $K_{3}>0$, i.e., $\lambda_{j}$ and $S_{j}$ decrease with $Q_{i}$ where $Q_{i}\in[0,~\frac{K_{3}}{K_{4}}]$;
  \item When $Q_{j}<Q_{e}$ and $Q_{e}>\frac{\sum_{k=1,\ne{i,j}}^{N}\lambda_{k}{S_{k}}^{\frac{2}{\beta}}Q_{k}}{\sum_{k=1,\ne{j}}^{N}\lambda_{k}{S_{k}}^{\frac{2}{\beta}}}$, 
we have $K_{5}<0$, $K_{6}<0$ and $K_{3}>0$, i.e., $\lambda_{j}$ and $S_{j}$ increase with $Q_{i}$ where $Q_{i}\in[\frac{K_{3}}{K_{4}}, +\infty)$;
\item When 
    $Q_{e}<\frac{\sum_{k=1,\ne{i,j}}^{N}\lambda_{k}{S_{k}}^{\frac{2}{\beta}}Q_{k}}{\sum_{k=1,\ne{j}}^{N}\lambda_{k}{S_{k}}^{\frac{2}{\beta}}}$
    , we have $K_{5}<0$, $K_{6}<0$ and $K_{3}<0$, i.e., $\lambda_{j}$ and $S_{j}$ increase with $Q_{i}$ where $Q_{i}\in[0,+\infty)$.
\end{itemize}
\end{corollary}
\begin{proof}
Please refer to Appendix \ref{proofcorollary8}.
\end{proof}
\begin{remark}
To keep the target SDP, the tradeoffs of the parameters in two different tiers are summarized as two cases: \begin{itemize}
\item For tier $i$ with cache size $Q_{i}>Q_{e}$, the BS density and transmit power of this tier can be reduced when another tier's cache size increases;
\item For tier $i$ with cache size $Q_{i}<Q_{e}$, the BS density and transmit power of this tier need to be increased when another tier's cache size increases.
\end{itemize}
\end{remark}
When increasing the density of the tier with small cache size, the SDP decreases since the tier can only
serve fewer content requests and bring the strong interference to other tiers with larger cache sizes. Thus, the density of such a tier cannot be reduced by increasing another tier's cache size. Inversely, increasing both the density of the tier with small cache size and the cache size of another tier can keep the fixed SDP.
\section{Simulation and Numerical Results}\label{simulationresults}
In this section, simulation and numerical results are provided to evaluate the performance of the proposed caching strategy, and present the impacts and tradeoffs of the network parameters. We consider a two-tier cache-enabled HetNet where the BS densities \{$\lambda_{1}$, $\lambda_{2}$\}=\{$\frac{1}{\pi{500^2}}$, $\frac{5}{\pi{500^2}}$\}, the transmit powers \{$S_{1}$, $S_{2}$\}=\{$53$, $33$\} dBm and the normalized cache sizes $\{Q_{1},~Q_{2}\}$=$\{40,~10\}$. We consider the interference-limited scenario and set the number of contents $M=200$, the path loss exponent $\beta=4$, and the SINR threshold $\tau=-10$ dB. In the simulation, the BSs are deployed according to independent HPPPs in an area of 5000 $m$$\times$5000 $m$, the content popularity distribution is modeled as the Zipf distribution, i.e., the popularity of the $j$-th ranked content is given by $t_{j}={\frac{1/j^{\gamma}}{\sum_{k=1}^{M}1/k^{\gamma}}}$,
where $\gamma\geq 0$ characterizes the skewness of the popularity distribution. The parameter $\gamma$ is set to 0.8 in the simulation. The optimal caching strategy is obtained by using interior point method. These parameters will not change unless specified otherwise.
\subsection{Optimal Caching Strategy}
$\textbf{Comparsion with other caching strategies}$.
To investigate the performance with larger content library, we set $M=1000$ in Fig. \ref{3compare}. It can be seen that the optimal probabilistic caching strategy always outperforms the two considered baselines. The first baseline is popular caching strategy where each BS only caches the most popular contents, and the second is uniform caching strategy where each BS caches each content randomly with equal probabilities. The popular caching strategy can only perform as well as the optimal probabilistic caching with highly skewed content popularity, e.g., $\gamma > 1$. The reason for this is that the majority of user requests only focus on those very few popular contents and the optimal probabilistic caching tends to cache those most popular contents. The performance of uniform caching strategy is not affected by $\gamma$ since it caches each content with equal probabilities. The theoretical results for the proposed optimal caching scheme are also validated by simulation in Fig. \ref{3compare}.

\textbf{Optimal caching probabilities for different network parameters}.
In Fig. \ref{probability}, we consider four cases:
\begin{enumerate}
  \item Case $1$: $\{Q_{1},Q_{2}\}$=\{$30$, $5$\}, \{$\lambda_{1}$, $\lambda_{2}$\}=\{$\frac{1}{\pi{500^2}}$, $\frac{10}{\pi{500^2}}$\}, \{$S_{1}$, $S_{2}$\}=\{$43$, $33$\} dBm;
  \item Case 2: $\{Q_{1},Q_{2}\}$=\{$40$, $10$\}, \{$\lambda_{1}$, $\lambda_{2}$\}=\{$\frac{1}{\pi{500^2}}$, $\frac{10}{\pi{500^2}}$\}, \{$S_{1}$, $S_{2}$\}=\{$43$, $33$\} dBm;
  \item Case $3$: $\{Q_{1},Q_{2}\}$=\{$30$, $5$\}, \{$\lambda_{1}$, $\lambda_{2}$\}=\{$\frac{1}{\pi{500^2}}$, $\frac{5}{\pi{500^2}}$\}, \{$S_{1}$, $S_{2}$\}=\{$43$, $33$\} dBm;
  \item Case $4$: $\{Q_{1},Q_{2}\}$=\{$30$, $5$\}, \{$\lambda_{1}$, $\lambda_{2}$\}=\{$\frac{1}{\pi{500^2}}$, $\frac{5}{\pi{500^2}}$\}, \{$S_{1}$, $S_{2}$\}=\{$53$, $40$\} dBm.
\end{enumerate}
The content popularity based on Zipf's law is also illustrated in Fig. \ref{probability}. It can be seen that the optimal caching probabilities highly depend on the network parameters, e.g., cache size, content popularity, transmit power and BS density. Caching the most popular contents is not always the best solution. Specifically, the optimal caching probability is positively related to the content popularity. The results also show that increasing the BS cache size has a more significant impact on the optimal caching probabilities than increasing the BS density and transmit power.
\subsection{Impacts of network parameters with optimal cache scheme}

\begin{figure}
 \vspace{-2.8mm}
\begin{minipage}[t]{0.5\textwidth}
\centering
\includegraphics[width=3.0in, height=2.1in]{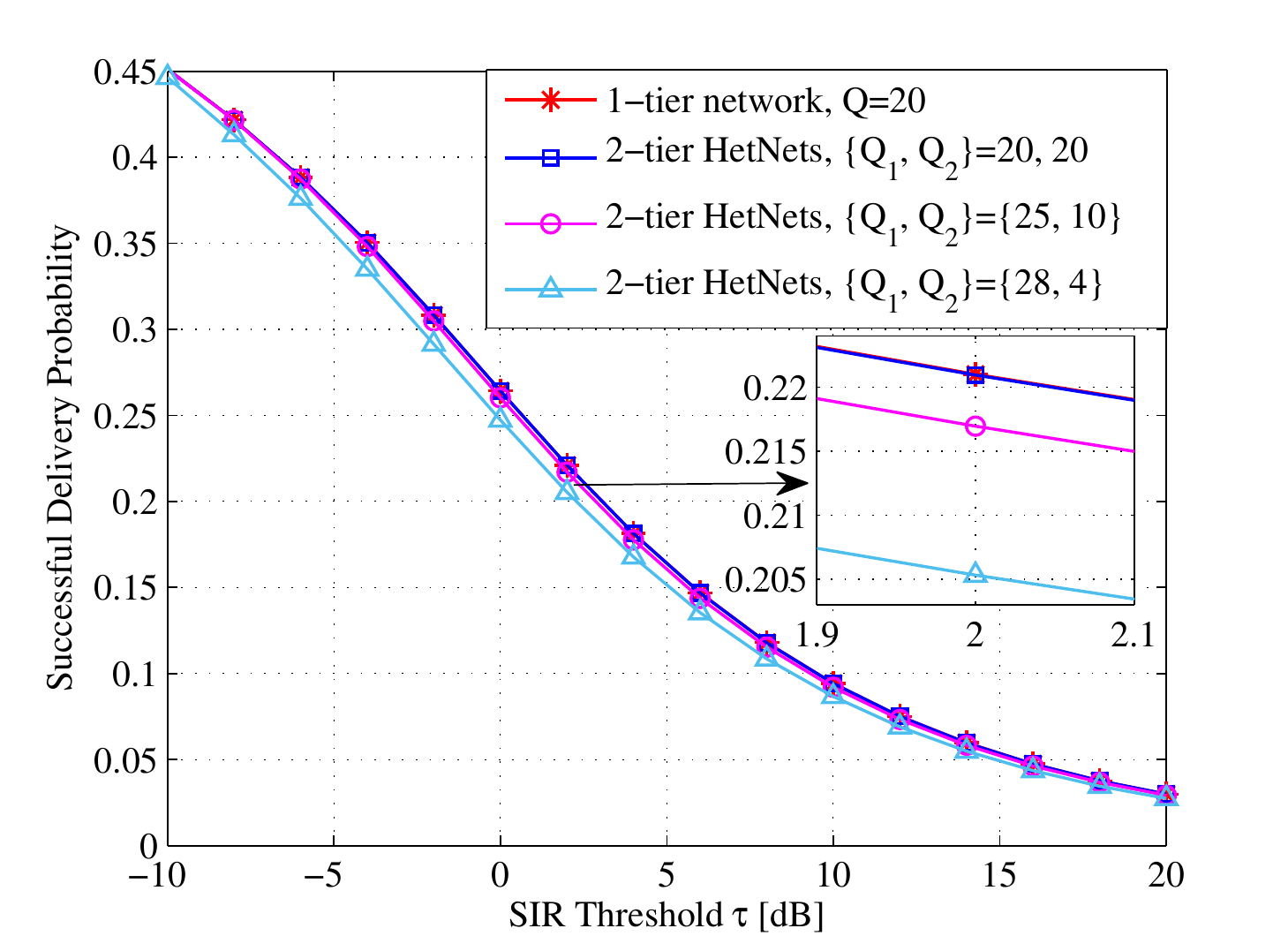}
 \vspace{-3.5mm}
 \caption{Comparison between 1-tier networks and 2-tier HetNets.}
\label{comparison}
 \end{minipage}
 \hspace{0.05in}
 \begin{minipage}[t]{0.5\textwidth}
 \centering
\includegraphics[width=3.0in, height=2.1in]{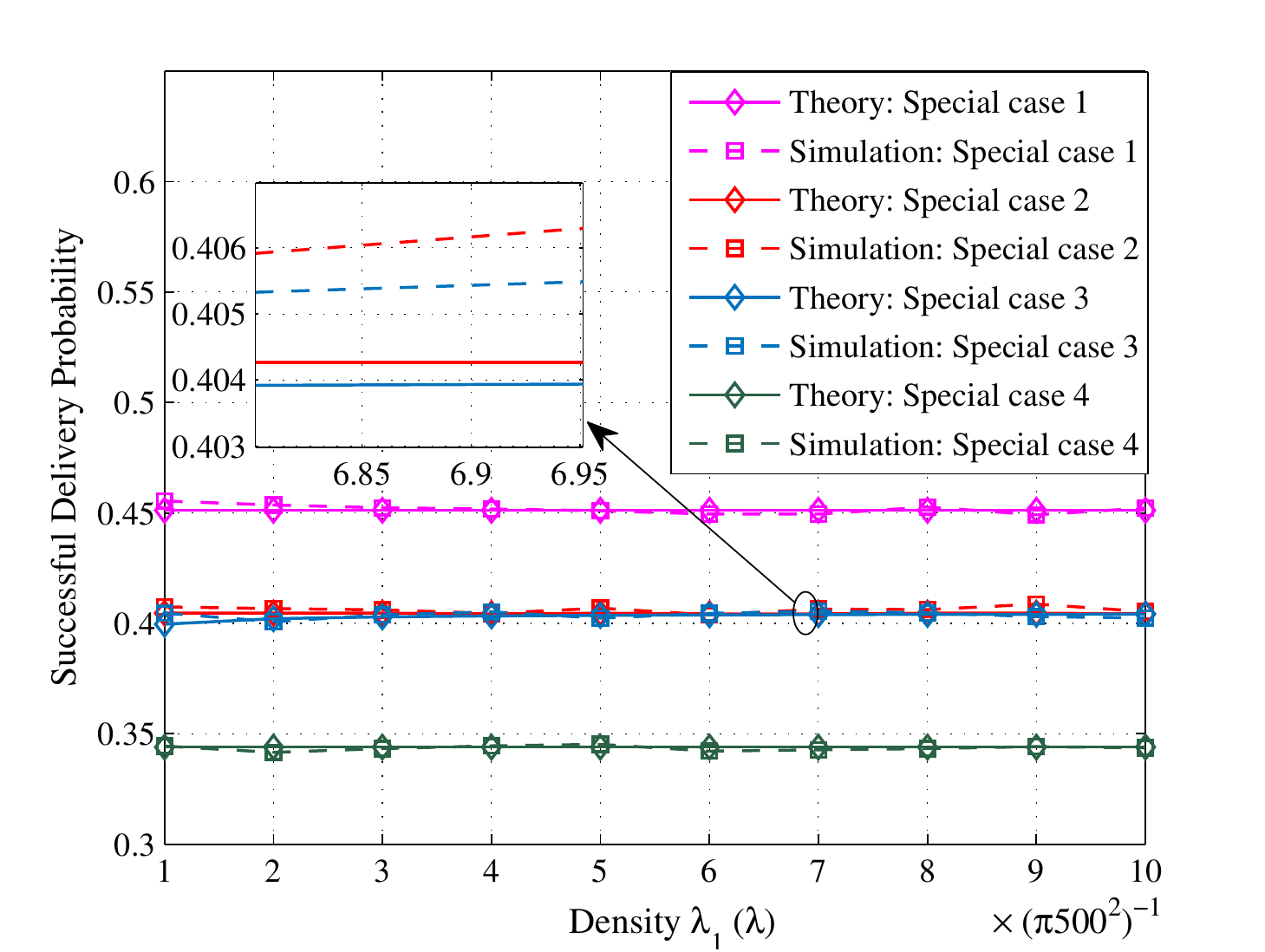}
\vspace{-3.5mm}
\caption{The impacts of tier-$1$ (single-tier network) BS density $\lambda_1$ ($\lambda$) on the maximum SDP $\mathcal{C}^{'*}$ for four special cases.}
\label{comparison2}
 \end{minipage}
  \vspace{-5.5mm}
\end{figure}
$\textbf{Impacts of cache size}$.
Fig. \ref{mathbf{Q}} illustrates that increasing the cache size can greatly improve the maximum SDP, which verifies the validity of Remark \ref{remark_1}. It also shows that increasing $\lambda_{1}$ or $S_{1}$ can further improve the optimal performance. That means more user requests can be served since BSs in the first tier have higher association probability and more cache resource.

$\textbf{Impacts of BS densities and transmit powers}$. It can be observed from Fig. \ref{lambda_{1}} that the optimal SDP decreases with the BS density $\lambda_{2}$ (with $Q_{2}=10$) and increases with $\lambda_{1}$ (with $Q_{1}=40$), which means deploying more BSs with smaller cache size decreases the optimal performance. This is expected because BSs with smaller cache size (e.g., pico or femto) can only provide less service to satisfy user requests while bringing strong interference to other BSs (e.g. macro or relay stations). In addition, the impacts of BS transmit powers are similar to that of BS densities, which are illustrated in Fig. \ref{S_{1}} and Fig. \ref{S_{2}}.
\subsection{The special cases}
$\textbf{Connection between 1-tier networks and $N$-tier HetNets}$. In Fig. \ref{comparison}, we consider four scenarios:

$1)$ 1-tier network with $Q_{e}=20$;

$2)$ 2-tier HetNets with $\{Q_{1},~Q_{2}\}=\{20,~20\}$;

$3)$ 2-tier HetNets with $\{Q_{1},~Q_{2}\}=\{25,~10\}$;

$4)$ 2-tier HetNets with $\{Q_{1},~Q_{2}\}=\{28,~4\}$.\\
In scenarios $2-4$, the two-tier BS cache sizes satisfy $\frac{\sum_{i=1}^{2}\lambda_{i}{S_{i}}^{\frac{2}{\beta}}Q_{i}}{\sum_{i=1}^{2}\lambda_{i}{S_{i}}^{\frac{2}{\beta}}}=Q_{e}=20$.
It can be seen that the optimal performances of the two-tier HetNets considered in scenarios $2-4$ are always upper bounded by that of the single-tier network considered in scenario $1$. In addition, the two-tier HetNet of scenario $2$ achieves the same performance as the single-tier network of scenario $1$ when their cache sizes satisfy $Q_{1}=Q_{2}=Q_{e}$. These results verify our conclusions in Proposition \ref{proposition_3} and Remark \ref{remark_3}.
\begin{figure*}
\vspace{-5mm}
\centering
\subfigure[The tradeoff between $\lambda_{2}$ and $Q_{1}$. Two cases: (1). \{$S_{1}$, $S_{2}$\}=\{43, 33\} dBm (2). \{$S_{1}$, $S_{2}$\}=\{53, 33\} dBm. We fix $\lambda_{1}$=$\frac{1}{\pi{500}^{2}}$ for both cases.]
{\label{trade3} 
\includegraphics[width=3.1in, height=2.3in]{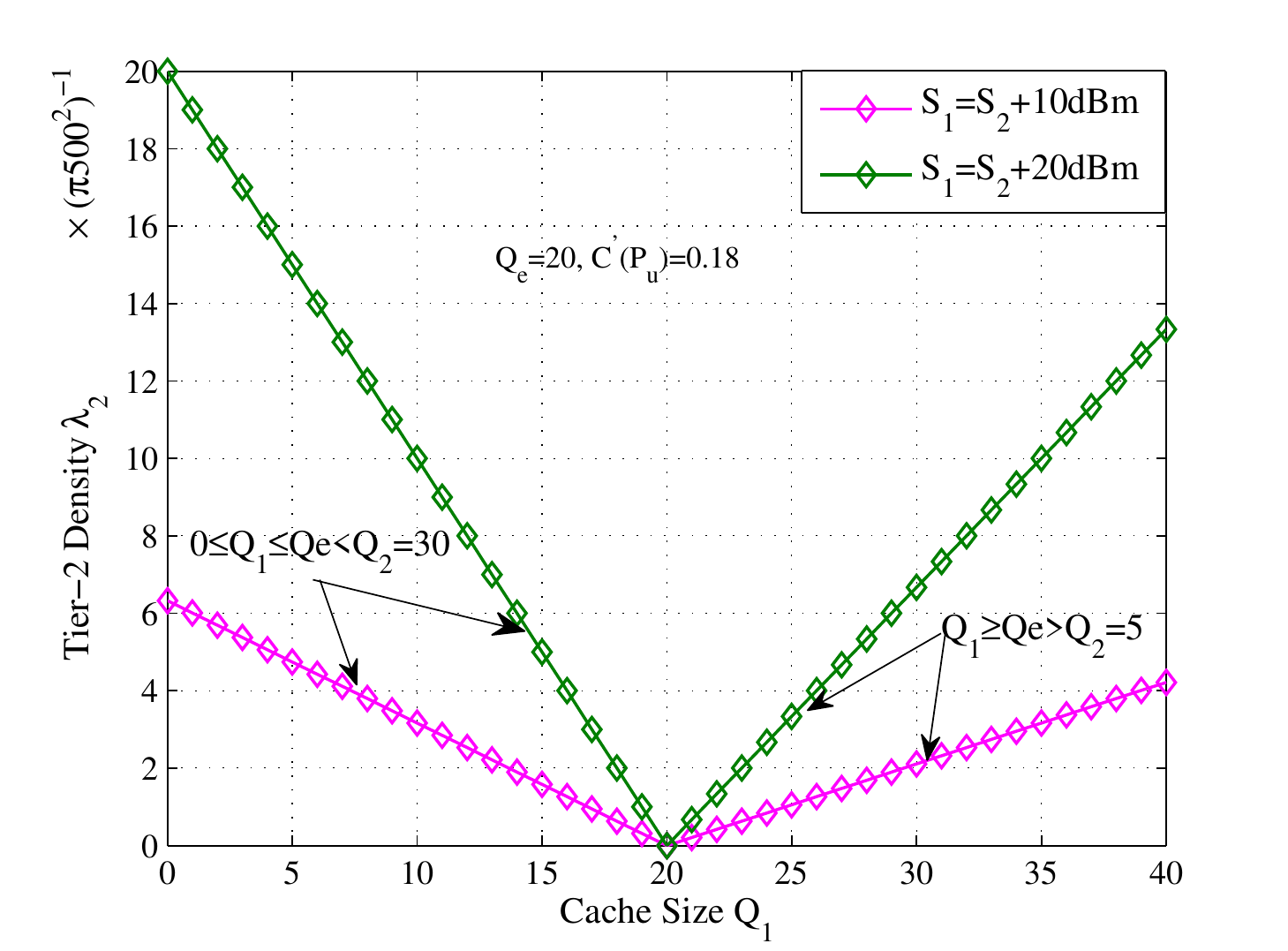}}
\hspace{-0.0in}
\subfigure[The tradeoff between $S_{2}$ and $Q_{1}$. Two cases: (1). \{$\lambda_{1}$, $\lambda_{2}$\}=\{$\frac{2}{\pi{500}^{2}}$, $\frac{10}{\pi{500}^{2}}$\} (2). \{$\lambda_{1}$, $\lambda_{2}$\}=\{$\frac{1}{\pi{500}^{2}}$, $\frac{10}{\pi{500}^{2}}$\}. We fix $S_{1}$=53 dBm for both cases.]
{\label{trade4} 
\includegraphics[width=3.1in, height=2.3in]{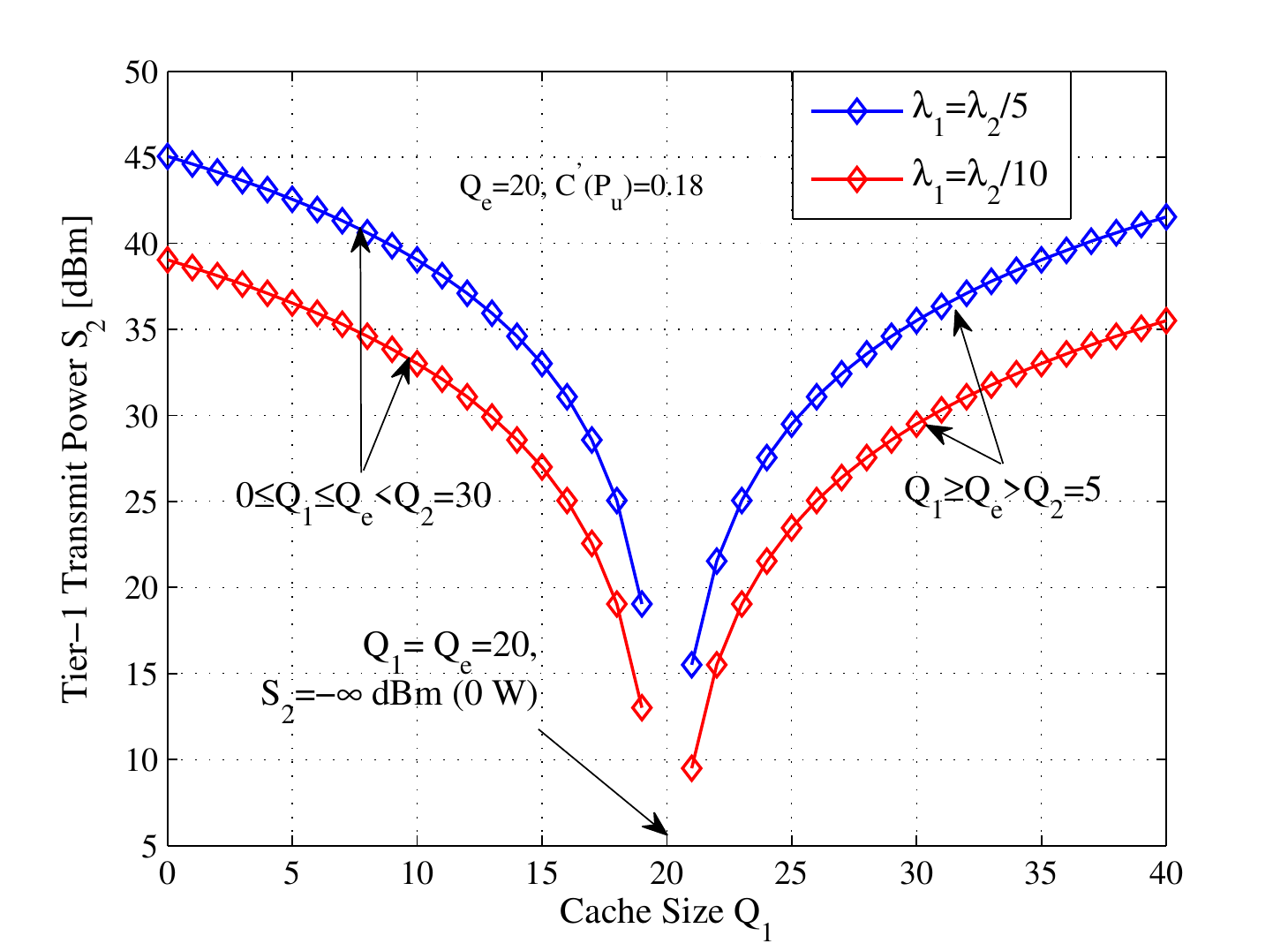}}
\vspace{-2mm}
\caption{Under fixed $\mathcal{C}^{'}(\mathbf{P_{u}})$, the tradeoffs of the BS density $\lambda_{2}$, transmit power $S_{2}$ and cache size $Q_{1}$ in a two-tier HetNet where $\mathcal{C}^{'}(\mathbf{P_{u}})=0.18$ ($Q_{e}$=20). By Corollary \ref{coro8}, we have $K_{3}>0$, $K_{5}>0$, $Q_{1}\in[0, Q_{e}]$ for $Q_{2}=30>Q_{e}$, and $K_{3}>0$, $K_{5}<0$, $Q_{1}\in(Q_{e},+\infty)$ for $Q_{2}=5<Q_{e}$.}
\label{tradeoff_3_4} 
\vspace{-4mm}
\end{figure*}

$\textbf{Impacts of BS density for special cases.}$ We consider four special scenarios in Fig. \ref{comparison2}:
\begin{enumerate}
\item 1-tier network with $Q=20$, $S=33$ dBm;
\item 2-tier HetNets with $Q_{1}=Q_{2}=15$, \{$S_{1}$, $S_{2}$\}=\{53, 33\} dBm, $\lambda_{2}=\frac{5}{\pi{500}^{2}}$;
\item 2-tier HetNets with $\{Q_{1},~Q_{2}\}=\{15,~5\}$, \{$S_{1}$, $S_{2}$\}=\{73, 33\} dBm ($S_{1}\gg S_{2}$),  $\lambda_{2}= \frac{5}{\pi{500}^{2}}$;
\item 2-tier HetNets with $\{Q_{1},~Q_{2}\}=\{15,~5\}$, $S_{1}$ = $S_{2}$ = 33 dBm, $\lambda_{2}$ = $\lambda_{1}$.
\end{enumerate}
We can see that the optimal performances are independent of the BS density in the considered four special cases, which also verifies our conclusions in Remark \ref{remark2} and Remark \ref{remark_3}. The optimal performance of scenario $2$ is almost the same as that of scenario $3$ because both them are almost entirely determined by the tier with cache size $Q=15$. Besides, from scenario $4$, it can be seen that when two tiers of BSs are deployed with the same density and transmit power, however, still have different cache sizes, performance independency on the BS density still exists.

\subsection{Tradeoffs of network parameters}
In Fig. \ref{tradeoff_3_4}, we present the tradeoffs of the network parameters in different tiers. To keep the fixed $\mathcal{C}^{'}(\mathbf{P_{u}})$, increasing the BS cache size $Q_{1}$ of the first tier cannot always bring the reduction of the BS density $\lambda_{2}$ or transmit power $S_{2}$ of the second tier. The crucial condition is that the cache size of the second tier should satisfy $Q_{2}>Q_{e}$. It can be further seen in Fig. \ref{trade3} that the slope that $\lambda_{2}$ varies with $Q_{1}$ is positively related to $S_{2}$. Similarly, the slope that $S_{2}$ varies with $Q_{1}$ is also positively related to $\lambda_{1}$, as illustrated in Fig. \ref{trade4}.

\section{Conclusion}\label{conclusion}
In this paper, we contributes to developing a general $N$-tier wireless cache-enabled HetNet framework and proposing an optimal probabilistic caching scheme for maximizing the SDP. Our analysis results show that the optimal performance of such a network not only depends on the cache size, but also the BS densities and transmit powers. Numerical results show that the proposed optimal caching strategy achieves a significant performance gain.

Toward analyzing the optimal performance of $N$-tier HetNets, we establish an interesting connection between $N$-tier HetNets and single-tier networks. Furthermore, we quantify the tradeoff between the BS cache size and the BS density with the uniform caching strategy. We show that: $i)$ increasing the BS caching capability does not always reduce the BS density; $ii)$ the BS density decreases with the BS caching capability when the BS cache size is larger than a threshold; $iii)$ the BS density is inversely proportional to the BS cache size within the same tier, and becomes a linear function of the BS cache sizes of different tiers.
\appendix
\subsection{Proof of Lemma \ref{lemma1}} \label{1}
By definition, we have
\begin{align}
\mathcal{W}_{i|j}&=\mathbb{P}\left[S_{r,i|j}>\max_{l,l\ne{i}}S_{r,l|j}\right]\nonumber\\
&=\mathbb{E}_{r_{i|j}}\left[\prod_{l=1,l\ne{i}}^{N}\mathbb{P}\left[S_{i}r^{-\beta}>S_{l}r_{l|j}^{-\beta}|r_{i|j}=r\right]\right]\nonumber\\
&=\int_{0}^{\infty}\prod_{l=1,l\ne{i}}^{N}\mathbb{P}\left[r_{l|j}>r\left(\frac{S_{l}}{S_{i}}\right)^{\frac{1}{\beta}}\right]f_{r_{i|j}}(r)\mathrm{d}{r},     \label{Wij}
\end{align}
where $f_{r_{i|j}}(r)$ is the probability density function (PDF) of $r_{i|j}$, and $r_{l|j}>r\left(\frac{S_{l}}{S_{i}}\right)^{\frac{1}{\beta}}$ means that it does not exists any BS of the $l$-th tier caching the content j in the area $A=\pi\left(r\left(\frac{S_{l}}{S_{i}}\right)^{\frac{1}{\beta}}\right)^2$. Using the fact that the null probability of a 2-D Poisson point process with density $\lambda$ in an area $A$ is $e^{-\lambda{A}}$, we can deduce
\begin{equation}
\mathbb{P}\left[r_{l|j}>r\left(\frac{S_{l}}{S_{i}}\right)^{\frac{1}{\beta}}\right] =e^{-\pi{p_{lj}}\lambda_{l}\left(\frac{S_{l}}{S_{i}}\right)^{\frac{2}{\beta}}r^{2}},
\label{rlj}
\end{equation}
Thus, $f_{r_{i|j}}(r)$ can be given by
\begin{equation}
f_{r_{i|j}}(r)=\frac{d\left[1-\mathbb{P}[{r_{i|j}>r}]\right]}{dr}=2\pi{p_{ij}}\lambda_{i}re^{-\pi\lambda_{i}p_{ij}r^2}  \label{fr}.
\end{equation}
Substituting (\ref{rlj}) and (\ref{fr}) into (\ref{Wij}), 
we have
\begin{align}
\mathcal{W}_{i|j}&=\int_{0}^{\infty}2\pi{p_{ij}}\lambda_{i}re^{-\pi\sum_{l=1,l\ne{i}}^{N}{p_{lj}}\lambda_{l}\left(\frac{S_{l}}{S_{i}}\right)^{\frac{2}{\beta}}r^{2}}e^{-\pi\lambda_{i}p_{ij}r^2}
\mathrm{d}{r} \nonumber\\ &\stackrel{(a)}{=}2\pi{p_{ij}}\lambda_{i}\int_{0}^{\infty}e^{-\pi\sum_{l=1}^{N}{p_{lj}}\lambda_{l}\left(\frac{S_{l}}{S_{i}}\right)^{\frac{2}{\beta}}r^{2}}r\mathrm{d}{r}\nonumber\\
&=\frac{\lambda_{i}p_{ij}{S_i}^{\frac{2}{\beta}}}{\sum_{l=1}^{N}{\lambda_{l}p_{lj}{S_l}^{\frac{2}{\beta}}}}, \label{Wij3}
\end{align}
where (a) follows that ${p_{lj}}\lambda_{l}\left(\frac{S_{l}}{S_{i}}\right)^{\frac{2}{\beta}}=\lambda_{i}p_{ij}$ for $l=i$. We complete the proof of Lemma \ref{lemma1}.
\subsection{Proof of Lemma \ref{lemma2}} \label{Lemma2}
According to the definition, we have
\begin{align}
f_{R_{i|j}}(r)&=\frac{d\left[1-\mathbb{P}[{R_{i|j}>r}]\right]}{dr}.\label{fR1}
\end{align}
Let $k_j$ denote the tier associated with the typical user requesting content $j$. Therefore, we have
\begin{align}
\mathbb{P}[R_{i|j}>r]&=\mathbb{P}[r_{i|j}>r|{k_j=i}] \nonumber\\ &=\frac{\mathbb{P}\left[r_{i|j}>r,S_{r,i|j}>\max_{l,l\ne{i}}S_{r,l|j}\right]}{\mathcal{W}_{i|j}} \nonumber \\
&{\setlength\arraycolsep{0.5pt}=}\frac{1}{\mathcal{W}_{i|j}}\int_{r}^{\infty}\prod_{l=1,l\ne{i}}^{N}\mathbb{P}\left[r_{l|j}{\setlength\arraycolsep{0.5pt}>}r\left(\frac{S_{l}}{S_{i}}\right)^{\frac{1}{\beta}}\right]f_{r_{i|j}}(r)\mathrm{d}{r}\nonumber \\
&\stackrel{(b)}{=}\frac{2\pi{p_{ij}}\lambda_{i}}{\mathcal{W}_{i|j}}\int_{r}^{\infty}e^{-\pi\sum_{l=1}^
{N}{p_{lj}}\lambda_{l}\left(\frac{S_{l}}{S_{i}}\right)^{\frac{2}{\beta}}r^{2}}r\mathrm{d}{r},\label{R>r}
\end{align}
where (b) is obtained by substituting (\ref{rlj}) and (\ref{fr}). Substituting (\ref{R>r}) into (\ref{fR1}), we thus have Lemma \ref{lemma2}.
\subsection{Proof of Proposition \ref{proposition1}} \label{2}
According to (\ref{coverageij}) and let 
$I_{l}=\sum_{k\in{\Phi_{l}}{\backslash}\{n_{io}\}}S_{l}|h_{l,k}|^2d_{l,k}^{-\beta}$, we then have
\begin{align}\label{Ci|j}
\mathcal{C}_{i|j}
&=\int_{0}^{\infty}\mathbb{P}\left[\text{SINR}_{i|j}(r)>\tau|r\right]f_{R_{i|j}}(r)\mathrm{d}{r}\nonumber\\
&=\int_{0}^{\infty}\mathbb{P}\left[|h_{i,o}|^2>\tau{S_{i}}^{-1}r^{\beta}(\sigma^2+\sum_{l=1}^{N}I_{l})|r\right]f_{R_{i|j}}(r)\mathrm{d}{r}\nonumber\\
&\stackrel{(c)}{=}\int_{0}^{\infty}\mathbb{E}_{I_{l}}\left[e^{-\tau{S_{i}}^{-1}r^{\beta}(\sigma^2+\sum_{l=1}^{N}I_{l})}|r,I_{l}\right]f_{R_{i|j}}(r)\mathrm{d}{r}\nonumber\\
&=\int_{0}^{\infty}\exp\left({-\tau{S_{i}}^{-1}r^{\beta}\sigma^2}\right)\prod_{l=1}^{N}\mathcal{L}_{I_{l}}\left(\tau{S_{i}}^{-1}r^{\beta}\right)f_{R_{i|j}}(r)\mathrm{d}{r},
\end{align}
where (c) follows $|h_{i,o}|^{2}\sim{\exp(1)}$. As discussed in Section \ref{performancemetric}, there exist two groups of interferences to $u_0$ from the $l$-th tier: the interference $I_{lj}$ from the BSs storing the requested content $j$ and the interference $I_{lj^{-}}$ from the BSs without storing content $j$. According to \cite{Martain}, the locations of these two groups of interfering BSs can be modeled as two thinned and independent HPPPs $\Phi_{lj}$ with density ${p_{lj}}\lambda_{l}$ and $\Phi_{lj^-}$ with density $(1-{p_{lj}})\lambda_{l}$, respectively.
By the condition $d_{l,k}\ge{}r_{l|j}\ge{}r\left(\frac{S_{l}}{S_{i}}\right)^{\frac{1}{\beta}}$ given in (\ref{Wij}), the Laplace transform $\mathcal{L}_{I_{lj}}$ can be calculated as
\begin{align}
&\mathcal{L}_{I_{lj}}\left(\tau{S_{i}}^{-1}r^{\beta}\right)=\mathbb{E}_{I_{lj}}\left[e^{-\tau{S_{i}}^{-1}r^{\beta}I_{lj}}\right]\nonumber\\
&=\mathbb{E}_{\Phi_{lj}}\left[\prod_{k\in{\Phi_{lj}}\backslash\{n_{io}\}}\frac{1}{1+\tau{S_{i}}^{-1}r^{\beta}S_{l}d_{l,k}^{-\beta}}\right]\nonumber\\
&\stackrel{(d)}{=}\exp\left({{-2\pi{p_{lj}}\lambda_{l}}\int_{r\left(\frac{S_{l}}{S_{i}}\right)^{\frac{1}{\beta}}}^{\infty}\!\left(\!1-\frac{1}{1+\tau{S_{i}}^{-1}r^{\beta}S_{l}x^{-\beta}}\right)\!x\mathrm{d}{x}}\!\right)\nonumber\\
&\stackrel{(e)}{=}\exp\left({-\pi{p_{lj}}\lambda_{l}\tau^{\frac{2}{\beta}}\left(\frac{S_{l}}{S_{i}}\right)^{\frac{2}{\beta}}r^{2}\frac{2}{\beta}\int_{\tau^{-1}}^{\infty}\left(
\frac{\mu^{\frac{2}{\beta}-1}}{1+\mu}\right)\mathrm{d}{\mu}}\right) \nonumber\\
&=\exp\left({-\pi{p_{lj}}\lambda_{l}(\frac{S_{l}}{S_{i}})^{\frac{2}{\beta}}r^{2}H(\tau,\beta)}\right),  \label{Elj}
\end{align}
where (d) follows from the probability generating functional (PGFL) of the PPP and (e) is obtained using $\mu=\tau^{-1}{S_{i}}r^{-\beta}{S_{l}^{-1}}x^{\beta}$ , $H(\tau,\beta)$ denotes $\frac{2\tau}{\beta-2}{}_2F_{1}(1,1-\frac{2}{\beta},2-\frac{2}{\beta},-\tau)$. Similarly, let $D(\tau,\beta)$ be $\frac{2}{\beta}\tau^{\frac{2}{\beta}}B(\frac{2}{\beta},1-\frac{2}{\beta})$, the Laplace transform $\mathcal{L}_{I_{lj^{-}}}$ is calculated as
\begin{align}
&\mathcal{L}_{I_{lj^{-}}}\left(\tau{S_{i}}^{-1}r^{\beta}\right)\nonumber\\
&=\exp\left({-(1-p_{lj})\lambda_{l}\int_{R^{2}}
\left(1-\frac{1}{1 + \tau{S_{i}}^{-1}r^{\beta}S_{l}x^{-\beta}} \right) x\mathrm{d}{x}}\right) \nonumber \\
&=\exp\left({-\pi(1-p_{lj})\lambda_{l}(\frac{S_{l}}{S_{i}})^{\frac{2}{\beta}}r^{2}D(\tau,\beta)}\right).  \label{Elj-}
\end{align}
Substituting (\ref{fR}), (\ref{Elj}), (\ref{Elj-}) into (\ref{Ci|j}), we have Proposition \ref{proposition1}.

\subsection{Proof of Proposition \ref{proposition2}}\label{3}
It is easy to see that the feasible set \{$\mathbf{P}|0\le{p_{ij}}\le{1}$, $\sum_{j=1}^{M}{p_{ij}}\le{Q_i}$, $~\forall{i}\in{\mathcal{N}},~\forall{j}\in{\mathcal{M}}$\} is a convex set. The objective function $\mathcal{C^{'}}(\mathbf{P})$ is a summation of $M$ polynomials. For $\forall{j}\in\mathcal{M}$, the $j$-th polynomial is given by
$F_{j}=\frac{\sum_{i=1}^{N}V_{ij}p_{ij}}{\sum_{i=1}^{N}G_{i}p_{ij}+E}$,
where $V_{ij}=\lambda_{i}{S_{i}}^{\frac{2}{\beta}}t_{j}$, $G_{i}=\lambda_{i}{S_{i}}^{\frac{2}{\beta}}T(\tau,\beta)$, and $E=D(\tau,\beta)\sum_{i=1}^{N}\lambda_{i}{S_{i}}^
{\frac{2}{\beta}}.$ The first derivative of $F_{j}$ is
\begin{equation}
\frac{\partial{F_{j}}}{\partial{p_{ij}}}=\frac{\partial{\mathcal{C}^{'}}}{\partial{p_{ij}}}=\frac{V_{ij}E}{\left(\sum_{i=1}^{N}G_{i}p_{ij}+E\right)^{2}}>0. \label{derivative}
\end{equation} We then have
$\frac{\partial^{2}{F_{j}}}{\partial{p_{ij}}\partial{p_{kj}}}=-ma_{i}a_{k} <0,$
where $m=\frac{2T(\tau,\beta)t_{j}E}{(\sum_{i=1}^{N}G_{i}p_{ij}+E)^3}$, $a_{i}=\lambda_{i}{S_i}^{\frac{2}{\beta}}$, and $a_{k}=\lambda_{k}{S_k}^{\frac{2}{\beta}}$.
Then, the Hessian matrix $\mathbf{H}_{N\times{N}}$ of $F_{j}$ is given by \begin{equation}\mathbf{H}_{N\times{N}}\triangleq{-m[a_{i}a_{k}]_{N\times{N}}=-m(\mathbf{a}\mathbf{a^{T}})},\end{equation}
where we define the column vector $\mathbf{a}\triangleq{(a_{1},a_{2},...,a_{N})^{T}}$. Obviously, the Hessian matrix $\mathbf{H}$ is a real symmetric negative semi-definite matrix, i.e., $F_{j}$ is concave of $p_{ij}$. According to the property that the summation of concave functions is also a concave function, $\mathcal{C^{'}}$ is also a concave function of $\mathbf{P}$. Thus, Problem $\mathrm{P1}$ is a concave optimization problem.

\subsection{Proof of Lemma \ref{lemma3}}\label{6}
By constructing the Lagrangian function of $\mathrm{P1}$, we have
\begin{align}
L(\mathbf{P},\boldsymbol{\mu},\boldsymbol{\lambda}
,\boldsymbol{\eta})=&\mathcal{C^{'}}(\mathbf{P})\!+\!\sum_{i=1}^{N}\sum_{j=1}^{M}\mu_{ij}p_{ij}\!+\!\sum_{i=1}^{N}\sum_{j=1}^{M}\lambda_{ij}(1-p_{ij})\nonumber\\
&+\sum_{i=1}^{N}\eta_{i}(Q_{i}-\sum_{j=i}^{M}p_{ij}), \label{L}
\end{align}
where
${\mu_{ij}}$, $\lambda_{ij}$ are the Lagrange multipliers associated with (\ref{condition1}), ${\eta_{i}}$ is the Lagrange multiplier associated with (\ref{condition2}), and we use $\boldsymbol{\mu}\triangleq{(\mu_{ij})_{N\times{M}}}$, $\boldsymbol{\lambda}\triangleq{(\lambda_{ij})_{N\times{M}}}$, $\boldsymbol{\eta}\triangleq{(\eta_{i})_{i\in{\mathcal{N}}}}$. Thus, we have
\begin{equation}
\frac{\partial{L(\mathbf{P}
,\boldsymbol{\mu},\boldsymbol{\lambda},\boldsymbol{\eta})}}{\partial{p_{ij}}}=\frac{V_{ij}E}{(\sum_{k=1}^{N}G_{k}p_{kj}+E)^2}
+\mu_{ij}-\lambda_{ij}-\eta_{i}\label{lange}.
\end{equation}
The KKT conditions are then written as follows
\begin{align}
\frac{\partial{L(\mathbf{P}^{*},\boldsymbol{\mu},\boldsymbol{\lambda}, \boldsymbol{\eta})}}{\partial{p^{*}_{ij}}}=0&,~~\forall{i}\in{\mathcal{N}}, \forall{j}\in{\mathcal{M}}, \label{K1}\\
\mu_{ij}p^{*}_{ij}=0,~~\lambda_{ij}(1-p^{*}_{ij})=0&,~~\forall{i}\in{\mathcal{N}}, \forall{j}\in{\mathcal{M}}, \label{K2}\\
\eta_{i}(Q_{i}-\sum_{j=1}^{M}p^{*}_{ij})=0&,~~\forall{i}\in{\mathcal{N}},\label{K3} \\
\sum_{j=1}^{M}p^{*}_{ij}=Q_{i},~0\le{p^{*}_{ij}}\le1&,~~\forall{i}\in{\mathcal{N}}, \forall{j}\in{\mathcal{M}}\label{K4}\\
\mu_{ij}\ge{0},~\lambda_{ij}\ge{0}&,~~\forall{i}\in{\mathcal{N}}, \forall{j}\in{\mathcal{M}}\label{K5}.
\end{align}
According to (\ref{lange}) and (\ref{K1}), we have that
\begin{equation}
\eta_{i}=\frac{V_{ij}E}{(\sum_{k=1}^{N}G_{k}p^{*}_{kj}+E)^2}+\mu_{ij}-\lambda_{ij}£¬~~~~\forall{k}\in{\mathcal{N}}, \forall{j}\in{\mathcal{M}}. \label{eta}
\end{equation}
As a result, we consider the following three cases to analyze the optimal solution:
\begin{itemize}
  \item $\eta_{i}\ge{\frac{V_{ij}E}{(\sum_{k=1,\ne{i}}^{N}G_{k}p^{*}_{kj}+E)^2}}:$ If $0<{p^{*}_{ij}}\le1$, then we have $0\le\lambda_{ij}<\mu_{ij}$, which is in conflict with $\mu_{ij}=0$ due to (\ref{K2}). If $p^{*}_{ij}=0$, we have $\lambda_{ij}=0$ and $\mu_{ij}\ge{0}$. Thus, we have $p^{*}_{ij}=0$.
  \item ${\frac{V_{ij}E}{(\sum_{k=1,\ne{i}}^{N}G_{k}p^{*}_{kj}+G_{i}+E)^2}}<\eta_{i}<{\frac{V_{ij}E}{(\sum_{k=1,\ne{i}}^{N}G_{k}p^{*}_{kj}+E)^2}}:$ If $0<{p^{*}_{ij}}<1$, we have $\mu_{ij}=\lambda_{ij}=0$ due to (\ref{K2}); If $p^{*}_{ij}=0$, we have $0\le\mu_{ij}<\lambda_{ij}$, which is in conflict with $\lambda_{ij}=0$ due to (\ref{K2}); If $p^{*}_{ij}=1$, we have $0\le\lambda_{ij}<\mu_{ij}$ and $\mu_{ij}=0$, which are contradictory. Thus, we have $0<{p^{*}_{ij}}<1$ and $p^{*}_{ij}=\frac{1}{G_{i}}\left(\sqrt{\frac{V_{ij}E}{\eta_{i}}}-\sum_{k=1,\ne i}^{N}G_{k}p^{*}_{kj}-E\right)$.
  \item $\eta_{i}\le{\frac{V_{ij}E}{(\sum_{k=1,\ne{i}}^{N}G_{k}p^{*}_{kj}+G_{i}+E)^2}}:$ Similar to the analysis for the first case, we have $p^{*}_{ij}=1$ and $\mu_{ij}=0$, $\lambda_{ij}\ge0$.
\end{itemize}
Summarizing above results, we thus have ($\ref{p1*}$). Based on (\ref{K4}), we have that $\eta_{i}$ satisfies $\sum_{j=1}^{M}p^{*}_{ij}=Q_{i}$ for $\forall{i}\in{\mathcal{N}}$. Thus, we complete the proof.
\subsection{Proof of Corollary \ref{corollary22}} \label{7}
First, we prove that $p_{j}^{*}$ decreases with the index $j$. For any two contents $i$, $j$ $\in\mathcal{M}$ and $i<j$, their popularities satisfy $t_{i}\geq t_{j}$. By (\ref{pj*}), there are only three possible cases: $p_{i}^{*}=p_{j}^{*}=1$, $p_{i}^{*}>p_{j}^{*}$, $p_{i}^{*}=p_{j}^{*}=0$, i.e., $p_{i}^{*}\ge{p_{j}^{*}}$. Thus, $p_{j}^{*}$ decreases with $j$, i.e., $p_{j}^{*}$ increases with $t_{j}$.

Then we prove that $p_{j}^{*}$ increases with the cache size $Q$. Based on (\ref{pj*}), without loss of generality, we assume $p_{i}^{*}=1$ for $1\le{i}\le{n}$ and $p_{j}^{*}=0$ for $m\le{j}\le{M}$ with $1\le{n}<m<M$. Then, we have
$\sum_{j=1}^{M}p^*_{j}=n+\sum_{j=n+1}^{m-1}p^*_{j}=n+\sum_{j=n+1}^{m-1}\left(\frac{1}{T(\tau,\beta)}{\sqrt{\frac{t_{j}D(\tau,\beta)}{\eta^{*}}}}-\frac{D(\tau,\beta)}{T(\tau,\beta)}\right)=Q.$ It can be seen that when $Q$ increases to $Q^{'}$, $\eta^*$ will decrease to $\eta^{*'}$ to satisfy $\sum_{j=1}^{M}p_{j}^{*'}=Q^{'}$, and $\frac{1}{T(\tau,\beta)}{\sqrt{\frac{t_{j}D(\tau,\beta)}{\eta^{*}}}}-\frac{D(\tau,\beta)}{T(\tau,\beta)}$ increases to $\frac{1}{T(\tau,\beta)}{\sqrt{\frac{t_{j}D(\tau,\beta)}{\eta^{*'}}}}-\frac{D(\tau,\beta)}{T(\tau,\beta)}$. Thus, we have $p_{j}^{*'}=p_{j}^{*}=1$ for $1\le{j}\le{n}$, $p_{j}^{*'}>p_{j}^{*}$ for $n+1\le{i}\le{m-1}$, and $p_{j}^{*'}=p_{j}^{*}=0$ or $p_{j}^{*'}>p_{j}^{*}=0$ for $m\le{j}\le{M}$. As a result, we have ${p_{j}^{*'}}\ge p_{j}^{*}$, i.e., $p_{j}^{*}$ increases with $Q$.

\subsection{Proof of Proposition \ref{proposition_3}}\label{4}
By letting $x_{j}=\frac{\sum_{i=1}^{N}\lambda_{i}S_{i}^{\frac{2}{\beta}}p_{ij}}{\sum_{i=1}^{N}\lambda_{i}S_{i}^{\frac{2}{\beta}}}$, it is clear that the objective function of $\mathrm{P1}$ is the same as that of $\mathrm{P2}$. Hence, to prove the first part of the proposition, it suffices to show that the feasible set of $\mathrm{P1}$ is a subset of that of $\mathrm{P2}$. Denote $\mathcal{P}=\{\mathbf{P}=[p_{ij}]_{N\times{M}}|~0\le{p_{ij}}\le{1},~\sum_{j=1}^{M}p_{ij}=Q_{i}$, $\forall{i}\in{\mathcal{N}},~\forall{j}\in{\mathcal{M}}\}$ as the feasible set of $\mathrm{P1}$ and $\mathcal{X}=\{\mathbf{x}=(x_{j})_{1\times{M}}|~0\le{x_{j}}\le{1}$, $\sum_{j=1}^{M}{x_{j}}=\frac{\sum_{i=1}^{N}\lambda_{i}S_{i}^{\frac{2}{\beta}}Q_{i}}{\sum_{i=1}^{N}\lambda_{i}S_{i}^{\frac{2}{\beta}}},~\forall{j}\in{\mathcal{M}}\}$
as the feasible set of $\mathrm{P2}$. Let $\mathbf{a}=(\frac{\lambda_{i}S_{i}^{\frac{2}{\beta}}}{\sum_{i=1}^{N}\lambda_{i}S_{i}^{\frac{2}{\beta}}})_{1\times{N}}$, then for any given element $\mathbf{P}\in\mathcal{P}$, we can define $\mathbf{x}=\mathbf{a}\mathbf{P}$. Since $\frac{\partial{x_{j}}}{\partial{p_{ij}}}=\frac{\lambda_{i}S_{i}^{\frac{2}{\beta}}}{\sum_{l=1}^{N}\lambda_{i}S_{i}^{\frac{2}{\beta}}}>0$ and $0\le{p_{ij}}\le{1}$ for $\forall{i}\in{\mathcal{N}},~\forall{j}\in{\mathcal{M}}$, we have $0\le{x_{j}}\le{1}$ for $\forall{j}\in{\mathcal{M}}$. Due to $\sum_{j=1}^{M}p_{ij}=Q_{i}$ for $i\in\mathcal{N}$, we have $\sum_{j=1}^{M}{x_{j}}=\frac{\sum_{i=1}^{N}\lambda_{i}S_{i}^{\frac{2}{\beta}}\sum_{j=1}^{M}p_{ij}}{\sum_{l=1}^{N}\lambda_{i}S_{i}^{\frac{2}{\beta}}}=\frac{\sum_{i=1}^{N}\lambda_{i}S_{i}^{\frac{2}{\beta}}Q_{i}}{\sum_{i=1}^{N}\lambda_{i}S_{i}^{\frac{2}{\beta}}}$.
Thus, we have $\mathbf{x}\in\mathcal{X}$. This means that for any $\mathbf{P}\in\mathcal{P}$, it always holds that $\mathbf{x}=\mathbf{a}\mathbf{P}\in\mathcal{X}$. Hence, the feasible set $\mathcal{X}$ of $\mathrm{P2}$ includes the transformation of the feasible set $\mathcal{P}$ of $\mathrm{P1}$ as a subset. As a result, the optimal value of $\mathrm{P2}$ provides an upper bound of that of $\mathrm{P1}$.

We next prove the second part of the proposition. If $Q_{i}=Q_e$ for $\forall{i}\in\mathcal{N}$, then for any given element $\mathbf{x}\in\mathcal{X}$, we can always construct an $\mathbf{P'}$ where $p_{ij}=x_{j}$ for $i\in\mathcal{N}$. Due to $0\le{p_{ij}}=x_{j}\le{1}$ for $\forall{i}\in{\mathcal{N}},~\forall{j}\in{\mathcal{M}}$ and $\sum_{j=1}^{M}p_{ij}=\sum_{j=1}^{M}x_{j}=Q_e$ for $\forall{i}\in{\mathcal{N}}$, we thus have $\mathbf{P'}\in\mathcal{P}$. This suggests that if $\mathbf{x}^{*}\in\mathcal{X}$ is an optimal solution of $\mathrm{P2}$, then $\mathbf{P'}$ with $p_{ij}=x^{*}_{j}$ is also an optimal solution of $\mathrm{P1}$. Hence, the optimal values of $\mathrm{P1}$ and $\mathrm{P2}$ are equal when $Q_i=Q_e$, $\forall{i}\in \mathcal{N}$. The proposition is thus proved.
\subsection{Proof of Theorem \ref{proposition5}} \label{proofpro5}
According to (\ref{C^{'}_{1e}}), when $\mathcal{C}^{'}_{1}(\mathbf{P_{1,u}})$ is maintained constant and $\lambda_{j}$, $S_{j}$ and $Q_{j}$ are fixed values for $\forall{j}\in\mathcal{N},j\ne{i}$, the equivalent cache size $Q_{e}$ also remains unchanged. We only need to consider the interactions of the $i$-th tier's parameters, thus we can derive that
\begin{equation} \lambda_{i}S_{i}^{\frac{2}{\beta}}=\frac{Q_{e}\sum_{j=1,\ne{i}}^{N}\lambda_{j}{S_{j}}^{\frac{2}{\beta}}
-\sum_{j=1,\ne{i}}^{N}\lambda_{j}{S_{j}}^{\frac{2}{\beta}}Q_{j}}{Q_{i}-Q_{e}}. \label{SDensity}
\end{equation}
Substituting (\ref{K_{2}}) into (\ref{SDensity}), (\ref{K_{33}}) into (\ref{SDensity}) respectively, we have (\ref{tradeoff1}) , (\ref{tradeoff2}).

\subsection{Proof of Corollary \ref{corollary7}} \label{proofcorollary7}
Equation (\ref{SDensity}) can be rewritten as following
\begin{equation} \lambda_{i}S_{i}^{\frac{2}{\beta}}=\frac{\left(Q_{e}-\frac{\sum_{j=1,\ne{i}}^{N}\lambda_{j}{S_{j}}^{\frac{2}{\beta}}Q_{j}}
{\sum_{j=1,\ne{i}}^{N}\lambda_{j}{S_{j}}^{\frac{2}{\beta}}}\right)
\sum_{j=1,\ne{i}}^{N}\lambda_{j}{S_{j}}^{\frac{2}{\beta}}}{Q_{i}-Q_{e}}.  \end{equation}

Obviously, when $Q_{e}>\frac{\sum_{j=1,\ne{i}}^{N}\lambda_{j}{S_{j}}^{\frac{2}{\beta}}Q_{j}}
{\sum_{j=1,\ne{i}}^{N}\lambda_{j}{S_{j}}^{\frac{2}{\beta}}}$, we have $K_{1}>0$. Based on $Q_{e}=\frac{\lambda_{i}{(S_{i})}^{\frac{2}{\beta}}Q_{i}+\sum_{j=1,\ne{i}}^{N}\lambda_{j}{S_{j}}^{\frac{2}{\beta}}Q_{j}}
{\lambda_{i}{(S_{i})}^{\frac{2}{\beta}}+\sum_{j=1,\ne{i}}^{N}\lambda_{j}{S_{j}}^{\frac{2}{\beta}}}$, $Q_{e}>\frac{\sum_{j=1,\ne{i}}^{N}\lambda_{j}{S_{j}}^{\frac{2}{\beta}}Q_{j}}
{\sum_{j=1,\ne{i}}^{N}\lambda_{j}{S_{j}}^{\frac{2}{\beta}}}$ is equivalent to
$Q_{i}>{\frac{\sum_{j=1,\ne{i}}^{N}\lambda_{j}{S_{j}}^{\frac{2}{\beta}}Q_{j}}{
\sum_{j=1,\ne{i}}^{N}\lambda_{j}{S_{j}}^{\frac{2}{\beta}}}}$. Meanwhile, because of $\lambda_{i}(S_{i})^{\frac{1}{\beta}}>0$, we have $Q_{i}>Q_{e}$. As a result, we have $K_{1}>0$ and $Q_{i}-Q_{e}>0$ when $Q_{i}\in(Q_{e},+\infty)$, then $K_{1}<0$ and $Q_{i}-Q_{e}<0$ when $Q_{i}\in[0,Q_{e})$. Because the relationship of $\lambda_{i}$ and $Q_{i}$ is inversely proportional, thus $\lambda_{i}$ decreases as $Q_{i}$ increases when $K_{1}>0$ and $Q_{i}-Q_{e}>0$, while $\lambda_{i}$ increases as $Q_{i}$ increases when $K_{1}<0$ and $Q_{i}-Q_{e}<0$. Similarly, we have $S_{i}$ decreases as $Q_{i}$ increases when $K_{2}>0$ and $Q_{i}-Q_{e}>0$, while $S_{i}$ increases as $Q_{i}$ increases when $K_{2}<0$ and $Q_{i}-Q_{e}<0$. Finally, we have the corollary.
\subsection{Proof of Corollary \ref{coro8}}\label{proofcorollary8}
According to (\ref{K_{3}})-(\ref{K_{6}}), we consider the following cases:
\begin{enumerate}
    \item  When $K_{3}>0$, from (\ref{K_{3}}), we have
         $Q_{e}>\frac{\sum_{k=1,\ne{i,j}}^{N}\lambda_{k}{S_{k}}^{\frac{2}{\beta}}Q_{k}}{\sum_{k=1,\ne{j}}^{N}\lambda_{k}{S_{k}}^{\frac{2}{\beta}}}$.
    \item When $K_{5}>0$ ($K_{6}>0$), i.e., $Q_{j}>Q_{e}$, then based on $Q_{e}=\frac{\sum_{k=1,\ne{j}}^{N}\lambda_{k}{S_{k}}^{\frac{2}{\beta}}Q_{k}+\lambda_{j}{S_{j}}^{\frac{2}{\beta}}Q_{j}}{\sum_{k=1,\ne{j}}^{N}\lambda_{k}{S_{k}}^{\frac{2}{\beta}}+\lambda_{j}{S_{j}}^{\frac{2}{\beta}}}$, we have $Q_{e}>\frac{\sum_{k=1,\ne{j}}^{N}\lambda_{k}{S_{k}}^{\frac{2}{\beta}}Q_{k}+\lambda_{j}{S_{j}}^{\frac{2}{\beta}}Q_{e}}{\sum_{k=1,\ne{j}}^{N}\lambda_{k}{S_{k}}^{\frac{2}{\beta}}+\lambda_{j}{S_{j}}^{\frac{2}{\beta}}}$, i.e., $Q_{e}>\frac{\sum_{k=1,\ne{j}}^{N}\lambda_{k}{S_{k}}^{\frac{2}{\beta}}Q_{k}}{\sum_{k=1,\ne{j}}^{N}\lambda_{k}{S_{k}}^{\frac{2}{\beta}}}>\frac{\sum_{k=1,\ne{i,j}}^{N}\lambda_{k}{S_{k}}^{\frac{2}{\beta}}Q_{k}}{\sum_{k=1,\ne{j}}^{N}\lambda_{k}{S_{k}}^{\frac{2}{\beta}}}$.
Further, we have $K_{3}>0$. Since $\lambda_{j}\ge0$ and $S_{j}\ge0$, we thus have $Q_{i}\in[0,~\frac{K_{3}}{K_{4}}]$.
    \item When $K_{5}<0$ ($K_{6}<0$) and $K_{3}>0$, i.e., $Q_{j}<Q_{e}$ and $Q_{e}>\frac{\sum_{k=1,\ne{i,j}}^{N}\lambda_{k}{S_{k}}^{\frac{2}{\beta}}Q_{k}}{\sum_{k=1,\ne{j}}^{N}\lambda_{k}{S_{k}}^{\frac{2}{\beta}}}$,
we thus have $Q_{i}\in[\frac{K_{3}}{K_{4}}, +\infty)$ due to $\lambda_{j}\ge0$ and $S_{j}\ge0$.
    \item When $K_{3}<0$, i.e., $Q_{e}<\frac{\sum_{k=1,\ne{i,j}}^{N}\lambda_{k}{S_{k}}^{\frac{2}{\beta}}Q_{k}}{\sum_{k=1,\ne{j}}^{N}\lambda_{k}{S_{k}}^{\frac{2}{\beta}}}$, then we have $Q_{e}<\frac{\sum_{k=1,\ne{j}}^{N}\lambda_{k}{S_{k}}^{\frac{2}{\beta}}Q_{k}}{\sum_{k=1,\ne{j}}^{N}\lambda_{k}{S_{k}}^{\frac{2}{\beta}}}$. Further, according to $Q_{e}=\frac{\sum_{k=1,\ne{j}}^{N}\lambda_{k}{S_{k}}^{\frac{2}{\beta}}Q_{k}+\lambda_{j}{S_{j}}^{\frac{2}{\beta}}Q_{j}}{\sum_{k=1,\ne{j}}^{N}\lambda_{k}{S_{k}}^{\frac{2}{\beta}}+\lambda_{j}{S_{j}}^{\frac{2}{\beta}}}$, we have $Q_{j}<Q_{e}$, i.e., $K_{5}<0$ ($K_{6}<0$).
        Thus, $Q_{i}\in[0,~+\infty)$.
  \end{enumerate}
We thus have this corollary.
\bibliographystyle{IEEEtran}
\bibliography{refer}
\begin{IEEEbiography}[{\includegraphics[width=1in,height=1.25in,clip,keepaspectratio]{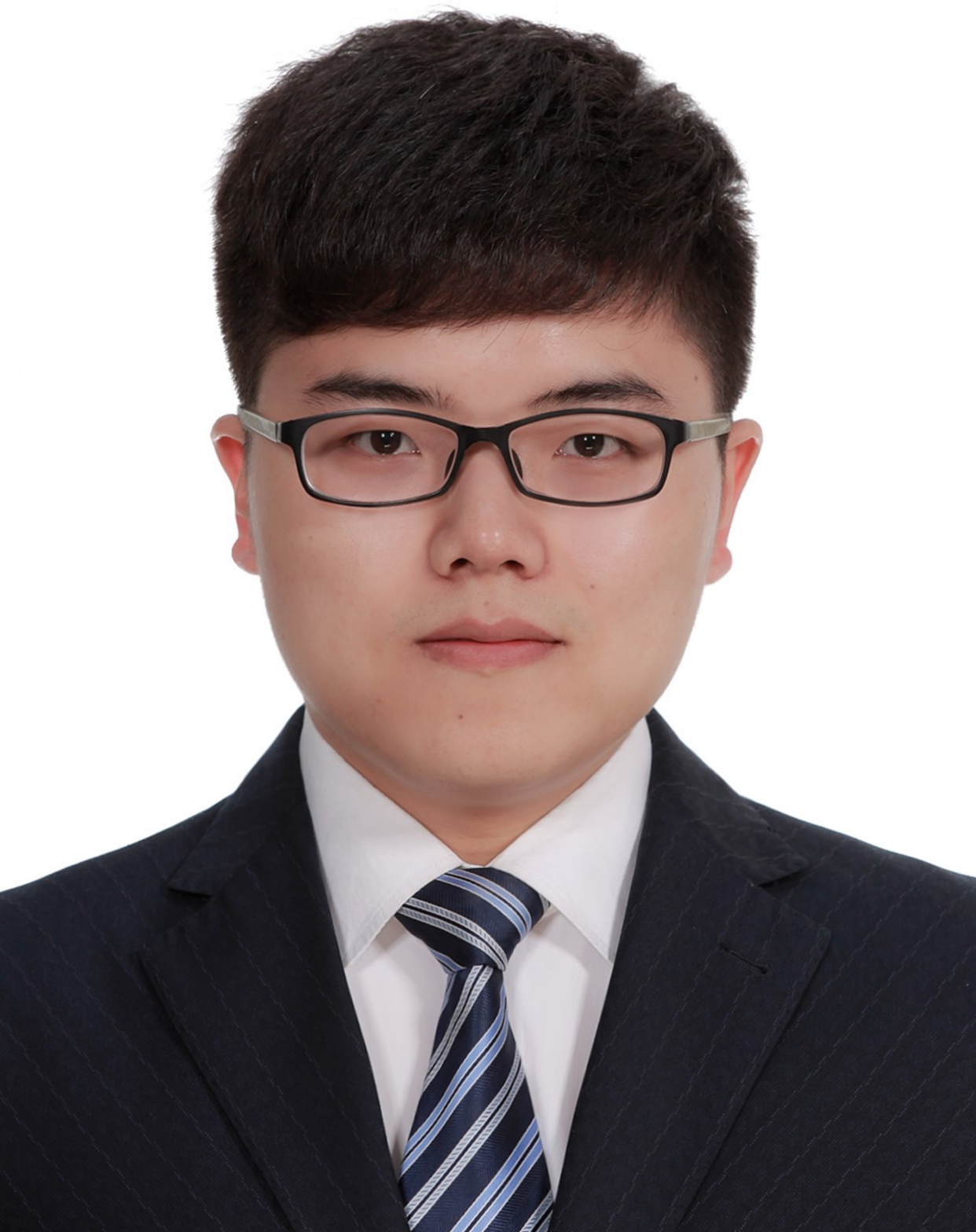}}]{Kuikui Li} received the B.E. degree in communications engineering from the School of Telecommunications Engineering, Xidian University, Xi'an, China, in 2016. He is currently pursuing the Ph.D. degree with the Department of Electric Engineering, Shanghai Jiao Tong University, China. His research interests include cache-enabled heterogeneous networks, mobile edge computing, and cooperative communications.
\end{IEEEbiography}
\begin{IEEEbiography}
[{\includegraphics[width=1in,height=1.25in,clip,keepaspectratio]{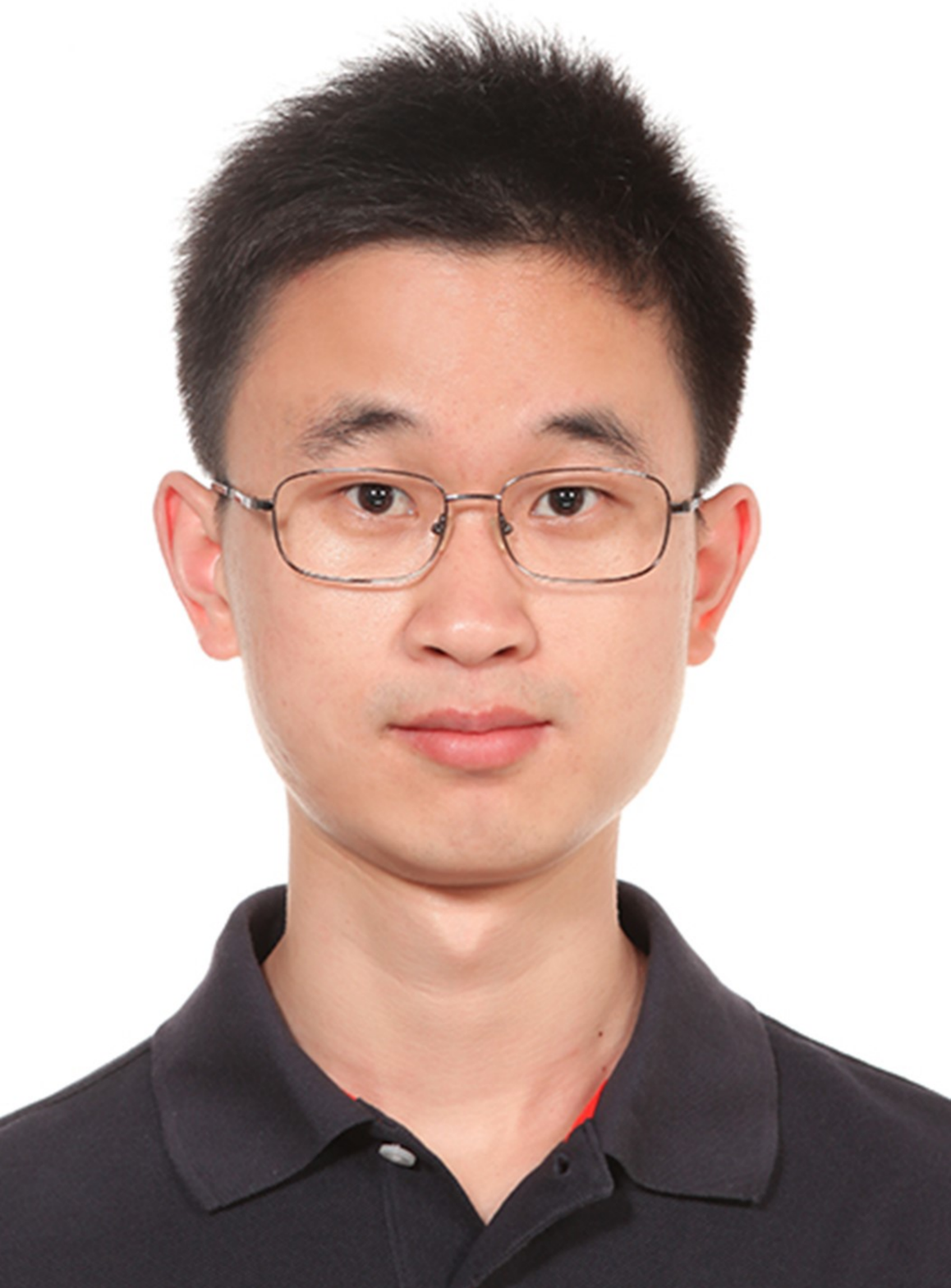}}]{Chenchen Yang}
received the B. Eng. degree in the School of Electronics and Information in 2013 from Northwestern Polytechnical University (NPU), Xi'an, China.  He is currently a Ph.D student in the Institute of Wireless Communications Technology (IWCT), Department of Electronic
Engineering, Shanghai Jiao Tong University (SJTU), Shanghai, China. From 2016 to 2017, he visited the Department of
Electrical Engineering, Columbia University, USA,
as a Research Scholar. His research expertise and interests include caching networks, cooperative communications and mobile edge computing.
\end{IEEEbiography}
\begin{IEEEbiography}
[{\includegraphics[width=1in,height=1.25in,clip,keepaspectratio]{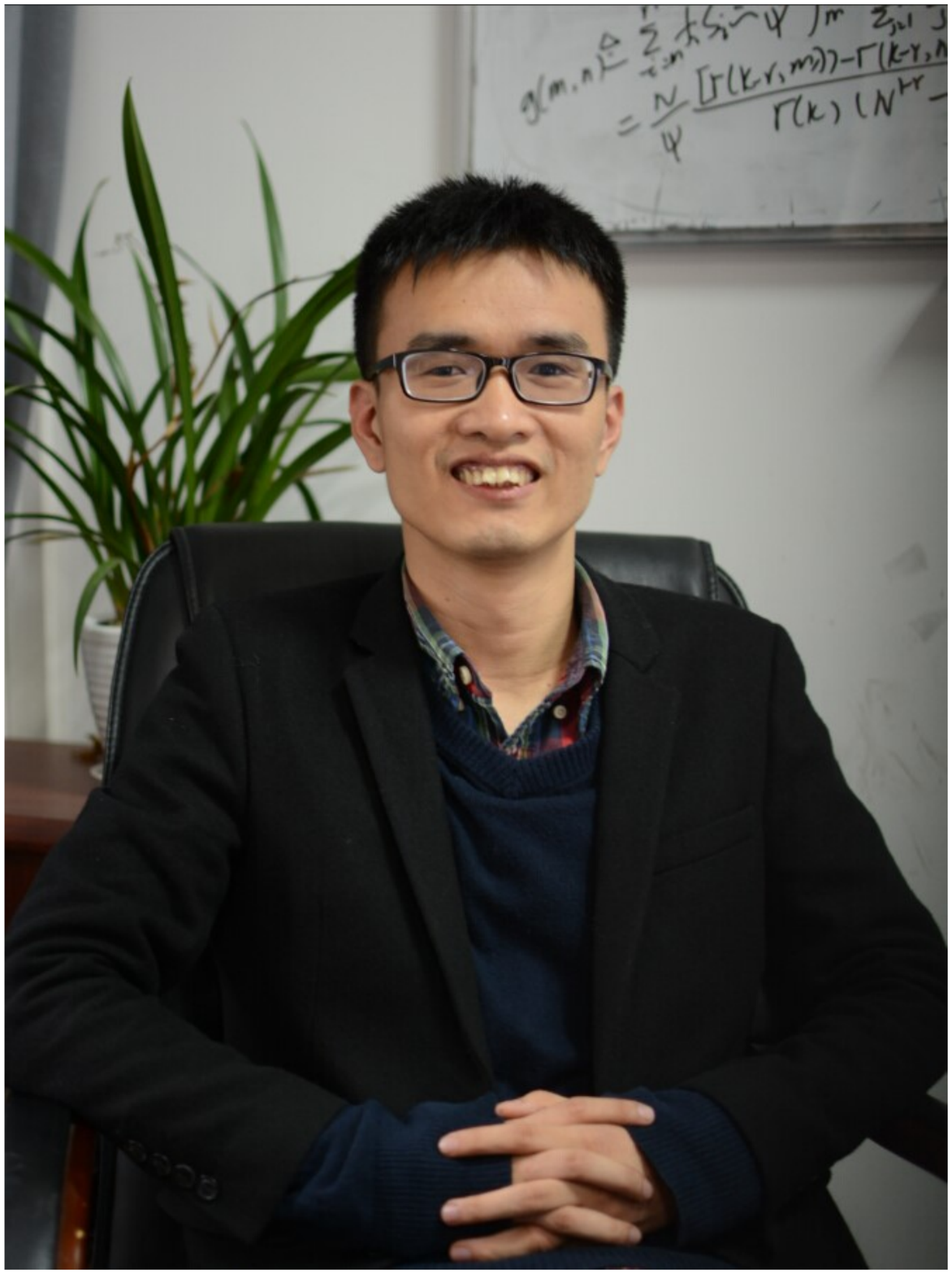}}]{Zhiyong Chen}
received the B.S. degree in electrical engineering from Fuzhou University, Fuzhou, China, and the Ph.D. degree from the School of Information and Communication Engineering, Beijing University of Posts and Telecommunications (BUPT), Beijing, China, in 2011. From 2009 to 2011, he was a visiting Ph.D. Student at the Department of Electronic Engineering, University of Washington, Seattle, USA. He is currently an Associate Professor with the Cooperative Medianet Innovation Center, Shanghai Jiao Tong University (SJTU), Shanghai, China. His research interests include mobile communications-computing-caching (3C) networks, mobile VR/AR delivery and mobile AI systems. He currently serves an Associate Editor of IEEE Access, and served as the Publicity Chair for the IEEE/CIC ICCC 2014 and a TPC member for major international conferences.
\end{IEEEbiography}
\begin{IEEEbiography}
[{\includegraphics[width=1in,height=1.25in,clip,keepaspectratio]{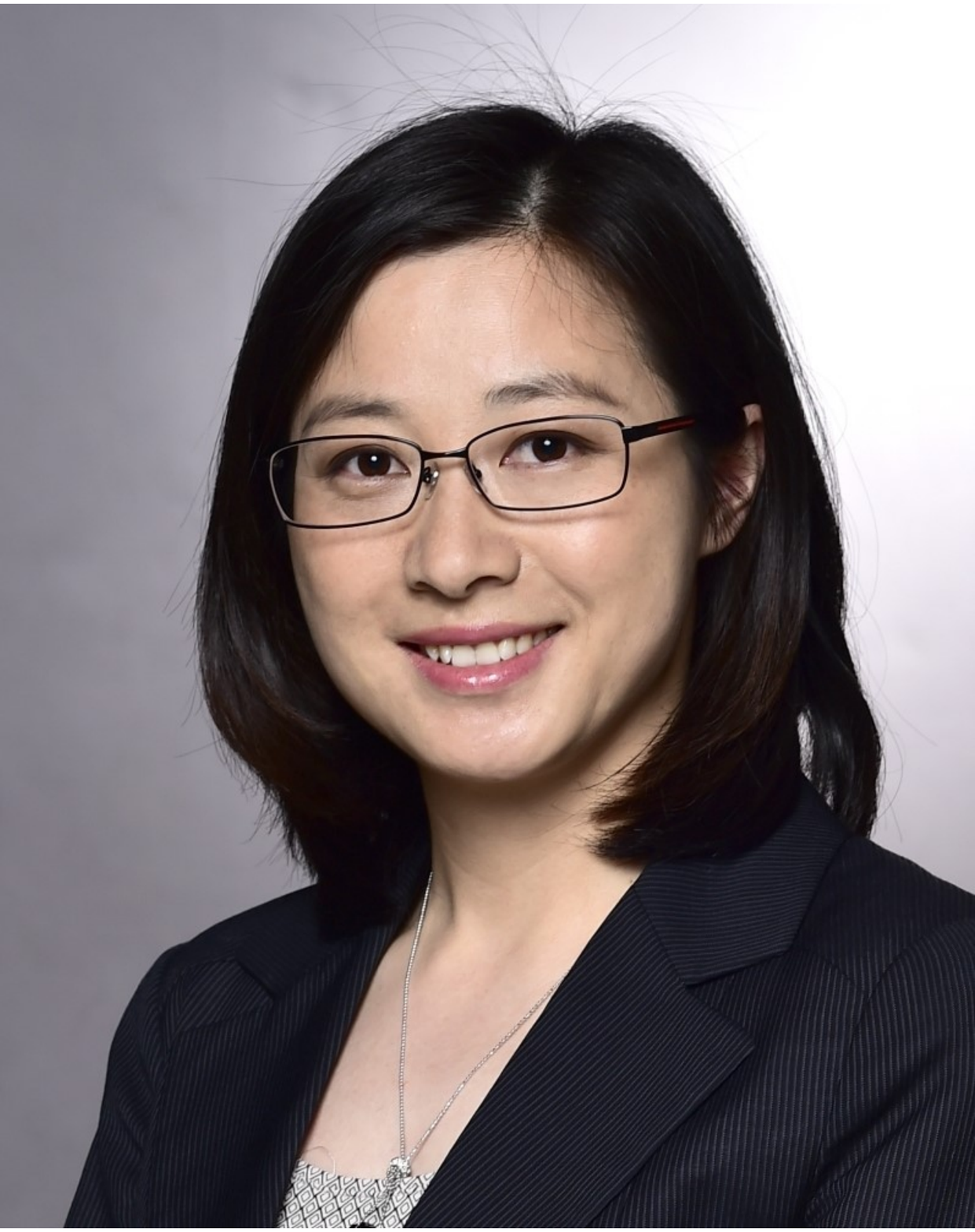}}]{Meixia Tao} (S'00-M'04-SM'10) received the B.S. degree in electronic engineering from Fudan University, Shanghai, China, in 1999, and the Ph.D. degree in electrical and electronic engineering from Hong Kong University of Science and Technology in 2003. She is currently a Professor with the Department of Electronic Engineering, Shanghai Jiao Tong University, China. Prior to that, she was a Member of Professional Staff at Hong Kong Applied Science and Technology Research Institute during 2003-2004, and a Teaching Fellow then an Assistant Professor at the Department of Electrical and Computer Engineering, National University of Singapore from 2004 to 2007. Her current research interests include wireless caching, physical layer multicasting, resource allocation, and interference management.

Dr. Tao currently serves as a member of the Executive Editorial Committee of the \textsc{IEEE Transactions on Wireless Communications} and an Editor for the \textsc{IEEE Transactions on Communications}. She was on the Editorial Board of the \textsc{IEEE Transactions on Wireless Communications} (2007-20110), the \textsc{IEEE Communications Letters} (2009-2012), and the \textsc{IEEE Wireless Communications Letters} (2011-2015). She also served as the TPC chair of IEEE/CIC ICCC 2014 and as Symposium Co-Chair of IEEE ICC 2015.

Dr. Tao is the recipient of the IEEE Heinrich Hertz Award for Best Communications Letters in 2013, the IEEE/CIC International Conference on Communications in China (ICCC) Best Paper Award in 2015, and the International Conference on Wireless Communications and Signal Processing (WCSP) Best Paper Award in 2012. She also receives the IEEE ComSoc Asia-Pacific Outstanding Young Researcher award in 2009.
\end{IEEEbiography}
\end{document}